\documentclass[aps,pra,onecolumn,10pt,superscriptaddress,nofootinbib]{revtex4-2}
\usepackage[utf8]{inputenc}
\usepackage[T1]{fontenc}
\usepackage{lmodern,microtype}
\usepackage{amsmath,amssymb,amsthm,mathtools,bm}
\usepackage[dvipsnames]{xcolor}
\usepackage{relsize,soul}
\usepackage{graphicx}
\usepackage{booktabs,multirow}
\usepackage{tikz}
\usetikzlibrary{quantikz2}
\usepackage{hyperref}
\hypersetup{colorlinks=true,citecolor=blue,linkcolor=blue,filecolor=blue,urlcolor=blue,breaklinks=true}
\usepackage[capitalise,noabbrev]{cleveref}

\newtheorem{theorem}{Theorem}
\newtheorem{lemma}[theorem]{Lemma}
\newtheorem{corollary}[theorem]{Corollary}
\newtheorem{proposition}[theorem]{Proposition}
\theoremstyle{definition}

\theoremstyle{remark}
\newtheorem{remark}{Remark}

\mathchardef\ordinarycolon\mathcode`\:
\mathcode`\:=\string"8000
\def\vcentcolon{\mathrel{\mathop\ordinarycolon}}
\begingroup \catcode`\:=\active
  \lowercase{\endgroup
  \let :\vcentcolon
  }
\newcommand{\nc}{\newcommand}

\renewcommand{\bra}[1]{\langle#1|}
\renewcommand{\ket}[1]{|#1\rangle}
\nc{\ketbra}[2]{\lvert#1\rangle\!\langle#2\rvert}
\nc{\braandket}[3]{\langle #1|#2|#3\rangle}
\nc{\kett}[1]{|#1\rangle\!\rangle}
\nc{\braakett}[2]{\langle\!\langle #1 | #2 \rangle\!\rangle}

\DeclarePairedDelimiter{\norm}{\lVert}{\rVert}
\DeclarePairedDelimiter{\abs}{\lvert}{\rvert}
\DeclarePairedDelimiter{\ceil}{\lceil}{\rceil}
\DeclarePairedDelimiter{\floor}{\lfloor}{\rfloor}

\nc{\RR}{{\mathbb R}}
\nc{\CC}{{\mathbb C}}
\nc{\NN}{{\mathbb N}}
\nc{\ZZ}{{\mathbb Z}}
\nc{\PP}{{\mathbb P}}
\nc{\cA}{{\mathcal A}}
\nc{\cB}{{\mathcal B}}
\nc{\cC}{{\mathcal C}}
\nc{\cD}{{\mathcal D}}
\nc{\cE}{{\mathcal E}}
\nc{\cF}{{\mathcal F}}
\nc{\cG}{{\mathcal G}}
\nc{\cH}{{\mathcal H}}
\nc{\cI}{{\mathcal I}}
\nc{\cJ}{{\mathcal J}}
\nc{\cK}{{\mathcal K}}
\nc{\cL}{{\mathcal L}}
\nc{\cM}{{\mathcal M}}
\nc{\cN}{{\mathcal N}}
\nc{\cO}{{\mathcal O}}
\nc{\cP}{{\mathcal P}}
\nc{\cQ}{{\mathcal Q}}
\nc{\cR}{{\mathcal R}}
\nc{\cS}{{\mathcal S}}
\nc{\cT}{{\mathcal T}}
\nc{\cU}{{\mathcal U}}
\nc{\cV}{{\mathcal V}}
\nc{\cW}{{\mathcal W}}
\nc{\cX}{{\mathcal X}}
\nc{\cY}{{\mathcal Y}}
\nc{\cZ}{{\mathcal Z}}

\nc{\tr}{\operatorname{Tr}}
\nc{\ox}{\otimes}
\nc{\trace}[2][]{\tr_{#1}\!\left[ #2 \right]}
\nc{\integral}[4]{\int_{#1}^{#2} {#3} \operatorname{d}\!{#4}}
\NewDocumentCommand{\ddx}{o m}{\IfNoValueTF{#1}
  {\frac{\operatorname{d}}{\operatorname{d}\!#2}}
  {\frac{\operatorname{d}^{\,#1}}{\operatorname{d}\!#2^{\,#1}}}}
\nc{\set}[1]{ \left\{{\kern .02083em} #1 {\kern .02083em}\right\} }
\nc{\setcond}[2]{ \left\{{\kern .02083em} #1 :{\kern .08333em} #2 {\kern .02083em}\right\} }
\nc{\poly}{\operatorname{poly}}
\nc{\BigO}[1]{{\mathcal{O}\!\left( #1 \right)}}
\nc{\BigTO}[1]{{\widetilde{\mathcal{O}}\!\left( #1 \right)}}
\nc{\BigOmega}[1]{{\Omega\!\left( #1 \right)}}
\nc{\BigTheta}[1]{{\Theta\!\left( #1 \right)}}

\nc{\bfx}{{\mathbf{x}}}
\nc{\bfz}{{\mathbf{z}}}
\nc{\bfX}{{\mathbf{X}}}
\nc{\bfn}{{\mathbf{n}}}
\nc{\bfm}{{\mathbf{m}}}
\nc{\bfL}{{\mathbf{L}}}
\nc{\bfp}{{\mathbf{p}}}
\nc{\bfU}{{\mathbf{U}}}

\nc{\prep}{{\mathtt{PREP}}}

\nc{\Cscell}[2][c]{\begin{tabular}[#1]{@{}c@{}}#2\end{tabular}}
\nc{\Lscell}[2][c]{\begin{tabular}[#1]{@{}l@{}}#2\end{tabular}}

\nc{\laurent}{\mathbb{P}_{1/2}}
\nc{\laurentint}{\mathbb{P}_{1}}

\ifdefined\hypersetup
\hypersetup{pdfcreator={},pdfproducer={}}
\else
\ifdefined{/Creator () /Producer ()}\fi
\fi
\ifdefined\pdfinfoomitdate\pdfinfoomitdate=1\fi
\ifdefined\pdfsuppressptexinfo\pdfsuppressptexinfo=-1\fi
\ifdefined\pdftrailerid\pdftrailerid{}\fi
\begin{document}

\title{Multivariate quantum signal processing with optimal query complexity}
\author{Mingrui Jing}
\email{mingruij0031@gmail.com}
\affiliation{QudeLeap Research, Shanghai 200030, China}
\affiliation{Thrust of Artificial Intelligence, Information Hub,\\
The Hong Kong University of Science and Technology (Guangzhou), Guangzhou 511453, China}
\author{Xin Wang}
\email{felixxinwang@hkust-gz.edu.cn}
\affiliation{Thrust of Artificial Intelligence, Information Hub,\\
The Hong Kong University of Science and Technology (Guangzhou), Guangzhou 511453, China}
\author{Zhan Yu}
\email{yu.zhan@u.nus.edu}
\affiliation{Centre for Quantum Technologies, National University of Singapore, Singapore}
\author{Lei Zhang}
\email{leizhang116.4@gmail.com}
\affiliation{QudeLeap Research, Shanghai 200030, China}
\affiliation{Thrust of Artificial Intelligence, Information Hub,\\
The Hong Kong University of Science and Technology (Guangzhou), Guangzhou 511453, China}

\begin{abstract}
Polynomial transformations of data encoded in quantum operations are a building block of quantum algorithms, and their query complexity measures how often those operations are used.
For multiple variables, existing methods either realize restricted polynomial families or implement monomials separately at a query cost exponential in the number of variables.
We present a quantum circuit that implements every multivariate trigonometric polynomial up to an explicit normalization.
The circuit processes one variable coherently over the frequency components of the remaining variables, allowing all components to share the same signal queries.
Each variable is queried exactly as many times as its degree in the polynomial, and we prove that this query complexity is optimal.
We further extend the construction to commuting unitaries.
The resulting circuit applies the multivariate polynomial simultaneously across their shared eigenbasis, using the minimum numbers of forward and inverse queries to each unitary.
We also study trainable rotations with fixed preparation and a local readout of the last input.
At query depth at least two, independent uniform initialization gives every active readout gradient a variance lower bound set by its squared branch weight and the inverse square root of the depth, uniformly over data states and dimensions.
For supervised squared loss, retaining sample correlations gives an initialization gradient-variance bound in terms of targets and input spectra.
This also bounds the expected loss decrease after the first exact gradient step with a specified step size.
These results provide a scalable and practical solution for developing multivariate quantum algorithms and quantum learning models.
\end{abstract}

\maketitle

\section{Introduction}\label{sec:introduction}

Polynomial transformations of a single input provide a common approach to many quantum algorithms~\cite{gilyen2019quantum, martyn2021grand}.
When data are encoded in a quantum operation, choosing a polynomial transformation of that operation gives a way to perform tasks such as Hamiltonian simulation~\cite{low2019hamiltonian, martyn2023efficient,childs2018toward}, linear-system solving and amplitude amplification~\cite{gilyen2019quantum}.
Many problems require several inputs and their joint dependence to be processed together.
For example, transformations of several variables can prepare quantum states with amplitudes specified by a multivariate function~\cite{mori2024efficient, rosenkranz2025quantum}.
Learning models with several inputs likewise use their joint dependence to represent multivariate functions~\cite{perezsalinas2020data, schuld2021effect}.

The goal is to realize general multivariate polynomials while minimizing the number of calls to each input operation.
These calls, or queries, measure the cost of accessing the input and can be a major resource when the operations are expensive to implement.
While polynomial transformations can use one query per unit of degree~\cite{low2016methodology, low2017optimal, yu2022power} for univariate input, existing multivariate constructions illustrate a tradeoff between this query efficiency and the range of polynomials they can realize.
One method retains a degree-matched query count for a restricted polynomial family~\cite{rossi2022multivariable}, and other methods obtain general polynomial representations by implementing and combining individual monomials~\cite{yu2024nonasymptotic, tang2026saqnn}; the queries in these constructions grow exponentially with the number of variables.
Although more variables require more coefficients to describe a general polynomial, whether multivariate polynomial transformation can be efficiently implemented in quantum circuits, remains an open problem.

In this work, we address this problem by presenting a circuit structure for simulating multivariate trigonometric polynomials up to a normalization factor, with optimal query complexity for each input.
Each variable is queried exactly as many times as its degree in the target polynomial, and we prove a matching lower bound for any circuit with the same access to the input.
The key is to share queries across the components of the polynomial.
We keep one variable as the input to a univariate polynomial transformation and organize the remaining variables by their frequency components.
Each component specifies a different univariate transformation, and the circuit carries out these transformations coherently with the same input queries.
This organization avoids the repeated queries of the individual-monomial constructions while retaining their general polynomial representation.
Table~\ref{tab:comparison scalar} compares the resulting query costs with previous constructions.

The same construction extends to commuting unitaries and preserves the optimal query guarantee for each input.
Commuting unitaries have a shared eigenbasis, so every common eigenvector carries one eigenphase from each input operation.
Our circuit applies the same multivariate polynomial to each of these eigenphase combinations, yielding a polynomial transformation of the input unitaries.
The numbers of forward and inverse queries to each unitary both attain their respective lower bounds.

The circuit also serves as a learning model with fixed preparation and trainable rotation angles.
For a fixed data state, the local readout studied here depends only on the last input unitary.
At query depth at least two, independent uniform initialization gives every active readout gradient a variance lower bound controlled by its branch weight and the query depth, uniformly over data states and dimensions.
The inverse-square-root depth dependence is sharp in the worst case.
For supervised squared loss, we retain the correlations between samples and bound the initialization gradient variance through the mean target and how targets correlate with the first spectral moment of the last input unitary.
This also gives a lower bound on the expected loss decrease after the first exact gradient update.
With a specified constant step size, exact gradient descent decreases the loss and the norm of its gradient tends to zero.
Convergence to zero training loss holds under an additional condition on the sample Jacobian throughout the trajectory.

\paragraph*{Note added.}
While completing this work, we became aware of concurrent work~\cite{low2026multivariate} on multivariate QSP that additionally accommodates noncommuting matrices.
Its main construction proceeds through the design and compilation
of analytic transfer functions, whereas ours directly converts
multivariate polynomials into explicit circuits, attaining optimal query complexity for the polynomial class considered here.

\section{Preliminaries}\label{sec:preliminaries}

This work concerns quantum circuits that implement a prescribed polynomial transformation of their inputs while accessing each input only through a fixed query gate.
The inputs are either $N$ real numbers or $N$ mutually commuting unitaries.
This section fixes the notation for multivariate polynomials, states the two query models, and reviews the univariate constructions on which the multivariate construction is built.

Throughout, $N$ denotes the number of variables, and $\NN$ denotes the set of nonnegative integers, which contains $0$.
The $j$-th variable is $x_j$, and the variables form the real vector $\bfx = (x_1, \ldots, x_N) \in \RR^N$.
For an $N$-dimensional array $C$ with finitely many nonzero entries $C_\bfn$, indexed by integer vectors $\bfn = (n_1, \ldots, n_N) \in \ZZ^N$, and for a function $z: \RR \to \CC$ of unit modulus, which below is one of the two phase functions $e^{ix/2}$ and $e^{ix}$, an \emph{$N$-variate polynomial} in $z$ is a function of the form
\begin{equation}\label{eqn:general polynomial}
    P(\bfx) = \sum_{\bfn} C_\bfn \prod_{j=1}^N z(x_j)^{n_j}
.\end{equation}
Three quantities of $P$ recur throughout.
$P$ is said to have \emph{parity} $\bfp = (p_1, \ldots, p_N) \in \set{0, 1}^N$ if $C_\bfn = 0$ whenever $\bfn \bmod 2 \neq \bfp$, where the modulus is taken entrywise.
The \emph{degree} of $P$ is the vector
\begin{equation}
    \deg(P) = \left(\max_{\bfn: C_\bfn \neq 0} \abs{n_1}, \ldots, \max_{\bfn: C_\bfn \neq 0} \abs{n_N}\right)
,\end{equation}
whose $j$-th entry $\deg_j(P)$ is the largest absolute exponent of $z(x_j)$ appearing in $P$; the degree of the zero polynomial is the zero vector.
Two norms measure the size of $P$: the coefficient $1$-norm $\norm{P}_{c, 1} = \sum_\bfn \abs{C_\bfn}$ and the supremum norm $\norm{P}_\infty = \sup_{\bfx} \abs{P(\bfx)}$.
For vectors we write $\bfn \preceq \bfL$ when $n_j \leq L_j$ for every $j$, so that $\deg(P) \preceq \bfL$ means that $P$ has degree at most $L_j$ in the $j$-th variable.
For the two choices of $z$ below, $\abs{z(x)} = 1$, so the triangle inequality gives $\norm{P}_\infty \leq \norm{P}_{c, 1}$, with equality for monomials; by the Cauchy--Schwarz inequality and Parseval's identity, for $P \neq 0$ the ratio $\norm{P}_{c, 1} / \norm{P}_\infty$ is at most the square root of the number of nonzero coefficients.

Two choices of $z$ are used.
For scalar signals we take $z(x) = e^{ix/2}$.
For $\bfL = (L_1, \ldots, L_N) \in \NN^N$ and $C \in \CC^{(2L_1 + 1) \times \cdots \times (2L_N + 1)}$, an \emph{$N$-variate trigonometric polynomial} is a function of the form
\begin{equation}\label{eqn:trig polynomial}
    P(\bfx) = \sum_{n_1 = -L_1}^{L_1} \cdots \sum_{n_N = -L_N}^{L_N} C_\bfn\, e^{i(\bfn \cdot \bfx)/2}
,\end{equation}
and the set of all such functions, over all $\bfL$ and $C$, is denoted by $\laurent$.
The half angle is inherited from the single-qubit rotation $R_z(x) = e^{-ixZ/2}$, whose diagonal entries are the monomials $e^{-ix/2}$ and $e^{ix/2}$, and it gives the parity a concrete meaning: a polynomial in $\laurent$ with parity $p_j$ in the $j$-th variable is $2\pi$-periodic in $x_j$ when $p_j = 0$ and changes sign under $x_j \mapsto x_j + 2\pi$ when $p_j = 1$.
For unitary inputs we take $z(x) = e^{ix}$ and denote by $\laurentint$ the set of $N$-variate Laurent polynomials in $e^{ix_1}, \ldots, e^{ix_N}$, that is, the functions of the form in Equation~\eqref{eqn:general polynomial} with $z(x) = e^{ix}$ and exponents of both signs.
Every $P \in \laurentint$ is $2\pi$-periodic in each variable.
The two sets are related by doubling the exponents: a polynomial in $\laurent$ with parity $0^N$ and degree $\bfL$ is the same function as a polynomial in $\laurentint$ with degree $\bfL/2$, and conversely, where $0^N$ denotes the zero vector.
For $P \in \laurentint$ we also record the largest exponents of each sign separately,
\begin{equation}\label{eqn:directional degree}
    \deg_j^{+}(P) = \max\set{0, \max_{\bfn: C_\bfn \neq 0} n_j},
    \qquad
    \deg_j^{-}(P) = -\min\set{0, \min_{\bfn: C_\bfn \neq 0} n_j}
,\end{equation}
where the inner maximum and minimum over the empty set are taken to be $0$, so that $\deg_j^{\pm}(0) = 0$ and $\deg_j(P) = \max\set{\deg_j^{+}(P), \deg_j^{-}(P)}$ for every $P$, and we write $\deg^{\pm}(P)$ for the vectors with these entries.

A scalar input $x_j$ enters a circuit only through the \emph{signal rotation} $R_z(x_j)$ acting on a designated qubit, and the \emph{query complexity} of the circuit in $x_j$ is the number of times $R_z(x_j)$ is applied; all other gates are fixed once the polynomial is fixed and may act on any register.
A circuit $W(\bfx)$ \emph{simulates} $P$ if $P(\bfx)$ is the matrix element of $W(\bfx)$ between the all-zero states of a designated set of qubits for every $\bfx \in \RR^N$, which we write as
\begin{equation}\label{eqn:block convention}
    W(\bfx) = \begin{bmatrix}
        P(\bfx) & \ldots \\
        \ldots & \ldots
    \end{bmatrix}
.\end{equation}
For unitary inputs, let $\bfU = (U_1, \ldots, U_N)$ be mutually commuting unitaries acting on a system $S$ of dimension $d = 2^n$.
They share an orthonormal eigenbasis $\set{\ket{\chi_k}}_{k=0}^{d-1}$, and we write $U_j = \sum_{k=0}^{d-1} e^{ix_{j, k}} \ketbra{\chi_k}{\chi_k}$, where $\bfX = (x_{j, k})$ is the $N \times d$ array of eigenphases and $\bfx_{:, k} = (x_{1, k}, \ldots, x_{N, k})$ collects the eigenphases of the $k$-th common eigenvector.
For $P \in \laurentint$, the \emph{polynomial transformation} of $\bfU$ applies $P$ to the eigenphases,
\begin{equation}\label{eqn:unitary transformation}
    P(\bfU) = \sum_{k=0}^{d-1} P(\bfx_{:, k}) \ketbra{\chi_k}{\chi_k}
    = \sum_{\bfn} C_\bfn \prod_{j=1}^N U_j^{n_j}
,\end{equation}
which is a polynomial in the unitaries and their inverses and therefore does not depend on the choice of the eigenphases modulo $2\pi$.
The same formula defines $P(\bfU)$ for $P \in \laurent$ with parity $0^N$, with $U_j^{n_j/2}$ in place of $U_j^{n_j}$, since every $n_j$ is then even; for any other parity the value $e^{ix_{j, k}/2}$ depends on the representative of $x_{j, k}$ and $P(\bfU)$ is not well defined.
In the unitary model the inputs are accessed only through the two controlled oracles
\begin{equation}\label{eqn:oracles}
    C_0(U_j) = \begin{bmatrix}
        U_j^\dagger & 0 \\
        0 & I_S
    \end{bmatrix},
    \qquad
    C_1(U_j) = \begin{bmatrix}
        I_S & 0 \\
        0 & U_j
    \end{bmatrix}
,\end{equation}
that is, controlled-$U_j^\dagger$ activated by the control state $\ket{0}$ and controlled-$U_j$ activated by $\ket{1}$, and the query complexity counts the uses of each oracle separately.
A circuit $V(\bfU)$ simulates $P(\bfU)$ if $P(\bfU)$ is the block of $V(\bfU)$ selected by the all-zero state of the ancillary qubits, for which we use the same notation as in Equation~\eqref{eqn:block convention}.

\subsection{Univariate case}

The rest of this section reviews how previous work simulates $P(x)$ and $P(U)$ for a single variable.
When $N = 1$, $\laurent = \CC[e^{ix/2}, e^{-ix/2}]$ is the set of Laurent polynomials in $e^{ix/2}$, parity and degree are scalars, and a polynomial in $\laurent$ with parity $0$ is an element of $\laurentint = \CC[e^{ix}, e^{-ix}]$.

The simulation of univariate polynomials goes back to Ref.~\cite{low2016methodology}, which showed that a sequence of interleaved single-qubit rotations and a signal gate carrying $x$ realizes polynomial transformations in $\CC[x]$.
Subsequent work developed several conventions of quantum signal processing (QSP)~\cite{gilyen2019quantum, haah2019product, chao2020finding, yu2022power}, which are summarized in Ref.~\cite{martyn2021grand} and Ref.~\cite{wang2023quantum}.
Our construction is built on the trigonometric variant~\cite{yu2022power}, which characterizes the single-qubit sequences in the Fourier basis as follows.

\begin{lemma}[Trigonometric QSP, \cite{yu2022power}]\label{lem:trig qsp}
Let $L \in \NN$ and $P, Q \in \laurent$.
There exist $\omega \in \RR$, $\theta = (\theta_0, \theta_1, \ldots, \theta_L) \in \RR^{L+1}$ and $\phi = (\phi_0, \phi_1, \ldots, \phi_L) \in \RR^{L+1}$ such that
\begin{equation}
    R_z(\omega) R_y(\theta_0) R_z(\phi_0) \prod_{l = 1}^L R_z(x) R_y(\theta_l) R_z(\phi_l) = \begin{bmatrix}
        P(x) & -Q(x)\\
        Q^*(x) & P^*(x)
    \end{bmatrix}
    \quad \text{for all } x \in \RR
\end{equation}
if and only if
\begin{enumerate}
    \item $\deg(P) \leq L$ and $\deg(Q) \leq L$,\label{item:degree constraint}
    \item $P$ and $Q$ have parity $L \bmod 2$,\label{item:parity constraint}
    \item $\abs{P(x)}^2 + \abs{Q(x)}^2 = 1$ for all $x \in \RR$.\label{item:norm constraint}
\end{enumerate}
\end{lemma}

Item~\ref{item:parity constraint} does not restrict the functions that can be realized: a Laurent polynomial in $e^{ix}$ of degree $L$ is an element of $\laurent$ with degree $2L$ and parity $0$, so every trigonometric polynomial bounded by one is covered by a sequence of length $2L$.
This is what distinguishes the trigonometric variant: it realizes general univariate transformations on a single qubit without the linear combination of unitaries (LCU)~\cite{childs2012hamiltonian} that other conventions need to remove their parity constraint.
When only $P$ is of interest, the complementary polynomial $Q$ exists whenever $\norm{P}_\infty \leq 1$~\cite{yu2022power}, and its phase is free in Item~\ref{item:norm constraint}: replacing $Q$ by $e^{i\gamma} Q$ conjugates the matrix in Lemma~\ref{lem:trig qsp} as $U(x) \mapsto R_z(-\gamma) U(x) R_z(\gamma)$, which leaves $P$ unchanged, so the choice $\gamma = \omega$ removes the leading rotation $R_z(\omega)$ and adds its angle to $\phi_L$, which saves one parameter.

\begin{corollary}\label{coro:trig qsp}
Let $L \in \NN$ and $P \in \laurent$ with $\deg(P) \leq L$, parity $L \bmod 2$ and $\norm{P}_\infty \leq 1$.
Then there exist $\theta \in \RR^{L+1}$ and $\phi \in \RR^{L+1}$ such that
\begin{equation}
    R_y(\theta_0) R_z(\phi_0) \prod_{l = 1}^L R_z(x) R_y(\theta_l) R_z(\phi_l) = \begin{bmatrix}
        P(x) & \ldots\\
        \ldots & \ldots
    \end{bmatrix}
    \quad \text{for all } x \in \RR
.\end{equation}
\end{corollary}

For a single unitary, quantum phase processing (QPP)~\cite{wang2023quantum} extends the trigonometric sequence from a scalar signal to the $d$ eigenphases of $U$.
Each query $R_z(x)$ is replaced by one of the oracles $C_0(U)$ and $C_1(U)$ in Equation~\eqref{eqn:oracles}, acting on the single qubit and $S$.
Restricted to the subspace spanned by $\ket{0}\ket{\chi_k}$ and $\ket{1}\ket{\chi_k}$, these oracles act as $e^{-ix_k/2} R_z(x_k)$ and $e^{ix_k/2} R_z(x_k)$, respectively, so a sequence of even length that alternates the two oracles runs the scalar sequence on every eigenphase simultaneously, and the scalar phases cancel.
For $P$ with parity $0$ and the same angles as in Lemma~\ref{lem:trig qsp}, the resulting circuit has $P(U)$ as the block selected by the single-qubit state $\ket{0}$.
Section~\ref{sec:unitary} uses the same mechanism for several commuting unitaries.
\section{Multivariate polynomial transformations}\label{sec:construction}

This section constructs circuits that simulate multivariate polynomial transformations of any (unknown) input $\bfx$ or $\bfU$ with the query counts prescribed by the degrees, and proves that these counts cannot be lowered.
For $\bfL \in \NN^N$, the polynomials are taken from the sets
\begin{align}
    \PP_{\bfL} &= \setcond{P \in \laurent}{\deg(P) \preceq \bfL \textrm{, $P$ has parity $\bfL \bmod 2$, } \norm{P}_{c, 1} \leq 1},\label{eqn:class scalar}\\
    \PP_{\bfL, 0} &= \setcond{P \in \laurent}{\deg(P) \preceq \bfL \textrm{, $P$ has parity $0^N$, } \norm{P}_{c, 1} \leq 1}\label{eqn:class even}
.\end{align}
The parity constraint in $\PP_\bfL$ is that of Lemma~\ref{lem:trig qsp}: a sequence with $L_j$ queries to $R_z(x_j)$ shifts the exponent of $e^{ix_j/2}$ by $\pm 1$ at each query and therefore produces exponents of parity $L_j \bmod 2$ only, which Proposition~\ref{prop:scalar lower bound} below makes precise for arbitrary circuits.
The set $\PP_{\bfL, 0}$ collects the polynomials of degree at most $\bfL$ that are $2\pi$-periodic in every variable, which are the polynomials with a well-defined action on unitaries; when every entry of $\bfL$ is even, $\PP_{\bfL, 0} = \PP_\bfL$.
Restricting to parity $0^N$ does not limit the functions that can be approximated, since Laurent polynomials in $e^{ix_j}$ with unconstrained coefficients are dense in the continuous $2\pi$-periodic functions~\cite{yu2024nonasymptotic, tang2026saqnn}; the bound $\norm{P}_{c, 1} \leq 1$ is a separate restriction, since a function bounded by one may only have approximants with larger coefficient norm, which are then simulated up to an explicit rescaling, as in the subnormalization of block encodings.
Section~\ref{sec:scalar} treats scalar signals and the set $\PP_\bfL$, Section~\ref{sec:unitary} treats commuting unitaries with prescribed positive and negative exponent ranges, including $\PP_{\bfL, 0}$ as the balanced even-degree case, and Section~\ref{sec:comparison} compares the results with previous work.

\subsection{Scalar signals}\label{sec:scalar}

\begin{figure}[t]
\centering
\resizebox{\linewidth}{!}{
\includegraphics{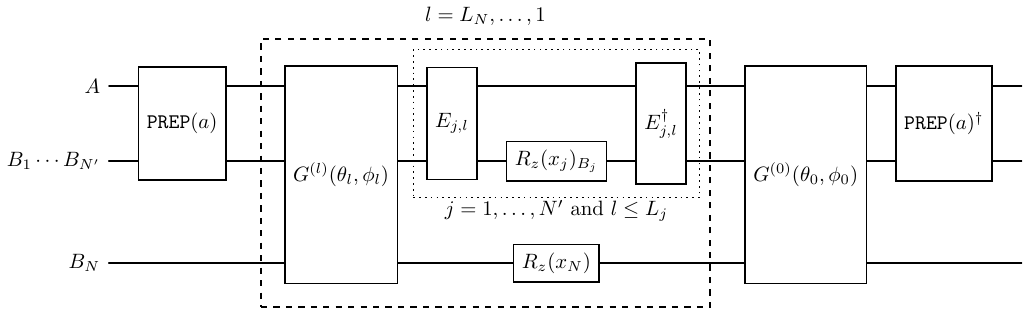}
}
\caption{The circuit $W(\bfx; a, \theta, \phi)$ in Equation~\eqref{eqn:scalar circuit} that simulates $P(\bfx)$ for $P \in \PP_\bfL$, where $a, \theta, \phi$ are parameters that depend on $P$. Here $N' = N - 1$ and $L_N$ is the largest entry of $\bfL$.}
\label{fig:mqsp single qubit}
\end{figure}

Fix $P \in \PP_\bfL$.
Without loss of generality, $L_N$ is the largest entry of $\bfL$; otherwise we relabel the variables.
Write $N' = N - 1$ and $\bfx' = (x_1, \ldots, x_{N'})$, and expand $P$ in the first $N'$ variables,
\begin{equation}\label{eqn:mu decomp}
    P(\bfx) = \sum_{\mu} e^{i(\mu \cdot \bfx')/2} P_\mu(x_N)
,\end{equation}
where $\mu = (n_1, \ldots, n_{N'})$ runs over the \emph{frequency grid} $\prod_{j=1}^{N'} \set{-L_j, -L_j + 2, \ldots, L_j}$ of the first $N'$ variables, which has $\ell = \prod_{j=1}^{N'}(L_j + 1)$ elements, and $P_\mu(x) = \sum_{k=-L_N}^{L_N} C_{\mu k}\, e^{ikx/2}$ is a univariate polynomial with $\deg(P_\mu) \leq L_N$ and parity $L_N \bmod 2$.
We denote by $\cI$ the set of frequency labels $\mu$ with $\norm{P_\mu}_{c, 1} > 0$ and abbreviate $\sum_{\mu \in \cI}$ as $\sum_\mu$; then $\sum_\mu \norm{P_\mu}_{c, 1} = \norm{P}_{c, 1}$.
The formulas below assume $P \neq 0$, so that $\norm{P}_{c, 1} > 0$ and $\cI$ is nonempty; the zero polynomial is treated separately in the proof of Theorem~\ref{thm:mqsp single qubit}.

The construction runs the univariate sequence of Corollary~\ref{coro:trig qsp} for all branches $P_\mu$ at once.
A register stores the label $\mu$ in superposition, with amplitudes set by the coefficient norms of the branches.
The sequence acts on one qubit and its rotation angles are controlled by $\mu$, so that every branch realizes its own polynomial while all branches share the same $L_N$ queries to $R_z(x_N)$.
The phase $e^{i(\mu \cdot \bfx')/2}$ of each branch is imprinted by $L_j$ queries to $R_z(x_j)$ for each $j < N$, each query conjugated by a fixed permutation of the register, and uncomputing the register at the end returns the weighted sum of the branches, which is $P(\bfx)$.

We use two registers $A$ and $B = B_1 \cdots B_N$.
The register $B$ consists of $N$ qubits, and the query $R_z(x_j)$ acts on $B_j$.
The register $A$ is auxiliary with $\abs{A} \geq \ceil{2\ell/\abs{B}}$ basis states, so that $AB'$ with $B' = B_1 \cdots B_{N'}$ has at least $\ell$ computational basis states; the smallest choice gives $\max\set{N, 1 + \ceil{\log_2 \ell}}$ qubits in total, which for equal degrees $L \geq 1$ is about $(N - 1)\log_2(L + 1) + 1$.
Each label $\mu$ of the frequency grid is encoded by a distinct computational basis state $\ket{\mu}$ of $AB'$, and $\Pi^\perp = I - \sum_\mu \ketbra{\mu}{\mu}$ denotes the projector onto the remaining states.
The circuit in Figure~\ref{fig:mqsp single qubit} consists of the following gates.
The preparation gate $\prep$ acts on $AB'$ and satisfies
\begin{equation}\label{eqn:prep}
    \prep \ket{0} = \sum_\mu a_\mu \ket{\mu}, \qquad a_\mu = \sqrt{\norm{P_\mu}_{c, 1} / \norm{P}_{c, 1}}
,\end{equation}
so that $\sum_\mu a_\mu^2 = 1$.
The rotation layers $G^{(l)}$ for $0 \leq l \leq L_N$ act on $AB$ and apply to $B_N$ a rotation controlled by the label,
\begin{equation}\label{eqn:rotation layer}
    G^{(l)}(\theta_l, \phi_l) = \sum_\mu \ketbra{\mu}{\mu} \ox R_y(\theta_{l, \mu}) R_z(\phi_{l, \mu}) + \Pi^\perp \ox I
,\end{equation}
where, for each $\mu \in \cI$, the angles $\theta_{l, \mu}, \phi_{l, \mu}$ with $0 \leq l \leq L_N$ are those given by Corollary~\ref{coro:trig qsp} for the normalized branch $\norm{P}_{c, 1} P_\mu / \norm{P_\mu}_{c, 1}$, whose supremum norm is at most $\norm{P}_{c, 1} \leq 1$.
The layers $G^{(l)}$ interleaved with $R_z(x_N)$ run one trigonometric QSP sequence for every label; this index-controlled form is the multi-controlled alternating phase sequence of Ref.~\cite{gilyen2019quantum}.
The permutation gates $E_{j, l}$ for $1 \leq j \leq N'$ and $1 \leq l \leq L_j$ act on $AB'$ and satisfy, for every label $\mu$ of the frequency grid,
\begin{equation}\label{eqn:permutation}
    \prod_{l=1}^{L_j} E_{j, l}^\dagger R_z(x_j)_{B_j} E_{j, l} \ket{\mu} = e^{i\mu_j x_j/2} \ket{\mu}
,\end{equation}
so that the $L_1 + \cdots + L_{N'}$ conjugated queries together imprint the phase $e^{i(\mu \cdot \bfx')/2}$ on $\ket{\mu}$.
Such permutations exist under the above condition $\abs{A}\, 2^{N'} \geq \ell$ on the register size.
Indeed, $\ket{\mu}$ acquires the phase $e^{ix_j/2}$ at layer $l$ if $E_{j, l}$ maps it to a basis state whose $B_j$ bit is $1$ and $e^{-ix_j/2}$ otherwise, so $\ket{\mu}$ has to be mapped to a bit-$1$ state in exactly $k(\mu) = (L_j + \mu_j)/2$ of the $L_j$ layers.
Pair each $\mu$ with $\mu_j \neq 0$ with the label $\bar\mu$ that negates its $j$-th entry, assign the layers $1, \ldots, k(\mu)$ to $\mu$ and the remaining $L_j - k(\mu) = k(\bar\mu)$ layers to $\bar\mu$, and pair the labels with $\mu_j = 0$ among themselves in the same way, the second member of each pair using the last $L_j/2$ layers; if one label with $\mu_j = 0$ remains unpaired, assign it any $L_j/2$ layers.
Then at every layer at most $\ceil{\ell/2}$ labels need a bit-$1$ state and at most $\ceil{\ell/2}$ need a bit-$0$ state, while $AB'$ contains $\abs{A}\, 2^{N'}/2 \geq \ceil{\ell/2}$ basis states of each kind, so a permutation with the required action exists at every layer.
The gates $E_{j, l}$ depend only on $\bfL$ and $N$ and carry no parameter of $P$.
Collecting the queries of layer $l$ into
\begin{equation}\label{eqn:scalar query layer}
    D_l(\bfx) = \left( \prod_{j:\, l \leq L_j} E_{j, l}^\dagger R_z(x_j)_{B_j} E_{j, l} \right) R_z(x_N)_{B_N}
,\end{equation}
where the product runs over $1 \leq j \leq N'$ with $l \leq L_j$, the circuit in Figure~\ref{fig:mqsp single qubit} reads
\begin{equation}\label{eqn:scalar circuit}
    W(\bfx; a, \theta, \phi) = \prep^\dagger\, G^{(0)} \left[\prod_{l=1}^{L_N} D_l(\bfx)\, G^{(l)}\right] \prep
,\end{equation}
with $\prep$ acting on $AB'$ and the identity understood on undisplayed registers.
Within this construction, we obtain the following theorem.

\begin{theorem}\label{thm:mqsp single qubit}
Let $\bfL \in \NN^N$, $P \in \PP_\bfL$ and $\ell = \prod_{j=1}^{N-1}(L_j + 1)$.
There exist $a \in \RR^{\ell}$ and $\theta, \phi \in \RR^{(L_N + 1) \times \ell}$ such that the circuit $W(\bfx; a, \theta, \phi)$ in Equation~\eqref{eqn:scalar circuit} satisfies
\begin{equation}
    W(\bfx; a, \theta, \phi) = \begin{bmatrix}
        P(\bfx) & \ldots \\
        \ldots & \ldots
    \end{bmatrix}
    \quad \text{for all } \bfx \in \RR^N
.\end{equation}
Moreover, $W(\bfx; a, \theta, \phi)$ uses exactly $L_j$ queries to $R_z(x_j)$ for each $j$, and at most $2(L_N + 2)\ell - 2$ single-parameter rotation gates whose parameters depend on $P$.
\end{theorem}
\begin{proof}
If $P = 0$, fix any label $\mu_0$ of the frequency grid, let $\prep$ map $\ket{0}$ to $\ket{\mu_0}$, keep the permutations $E_{j, l}$, take $G^{(l)}$ for $l \geq 1$ to be identities and let $G^{(0)}$ apply $R_y(\pi)$ to $B_N$ on the label $\mu_0$; since each $E_{j, l}$ permutes the computational basis, every conjugated query multiplies $\ket{\mu_0}$ by a phase, so $\ket{\mu_0} \ox \ket{0}_{B_N}$ is an eigenvector of every $D_l(\bfx)$, and $\bra{0} R_y(\pi) \ket{0} = 0$ makes the matrix element vanish.
Hence assume $P \neq 0$, so that $\norm{P}_{c, 1} > 0$ and $\cI \neq \emptyset$, and let $a$, $\theta$, $\phi$ be the parameters in Equations~\eqref{eqn:prep} and~\eqref{eqn:rotation layer}.
Corollary~\ref{coro:trig qsp} applies to each normalized branch, because $\deg(P_\mu) \leq L_N$, $P_\mu$ has parity $L_N \bmod 2$ and $\norm{P}_{c, 1} \norm{P_\mu}_\infty / \norm{P_\mu}_{c, 1} \leq \norm{P}_{c, 1} \leq 1$, so for every $\mu \in \cI$,
\begin{equation}\label{eqn:branch sequence}
    R_y(\theta_{0, \mu}) R_z(\phi_{0, \mu}) \prod_{l=1}^{L_N} R_z(x) R_y(\theta_{l, \mu}) R_z(\phi_{l, \mu}) \ket{0} = \frac{\norm{P}_{c, 1} P_\mu(x)}{\norm{P_\mu}_{c, 1}} \ket{0} + \ldots \ket{1}
.\end{equation}
On the label state $\ket{\mu}$, the layer $G^{(l)}$ acts on $B_N$ as $R_y(\theta_{l, \mu}) R_z(\phi_{l, \mu})$, and by Equation~\eqref{eqn:permutation} the conjugated queries in the layers $D_1(\bfx), \ldots, D_{L_N}(\bfx)$ multiply $\ket{\mu}$ by $e^{i(\mu \cdot \bfx')/2}$ in total.
Therefore
\begin{align}
    W(\bfx; a, \theta, \phi) \ket{0}_{AB}
    &= \prep^\dagger\, G^{(0)} \left[\prod_{l=1}^{L_N} D_l(\bfx)\, G^{(l)}\right] \sum_\mu a_\mu \ket{\mu} \ox \ket{0}_{B_N} \\
    &= \prep^\dagger \sum_\mu a_\mu\, e^{i(\mu \cdot \bfx')/2} \ket{\mu} \ox \left( \frac{\norm{P}_{c, 1} P_\mu(x_N)}{\norm{P_\mu}_{c, 1}} \ket{0} + \ldots \ket{1} \right)_{B_N} \\
    &= \ket{0}_{AB'} \ox \sum_\mu a_\mu^2\, e^{i(\mu \cdot \bfx')/2} \left( \frac{\norm{P}_{c, 1} P_\mu(x_N)}{\norm{P_\mu}_{c, 1}} \ket{0} + \ldots \ket{1} \right)_{B_N} + \ket{\mathord{\perp}} \\
    &= \ket{0}_{AB'} \ox \left( P(\bfx) \ket{0} + \ldots \ket{1} \right)_{B_N} + \ket{\mathord{\perp}}
,\end{align}
where $\bra{0}_{AB'} \ket{\mathord{\perp}} = 0$, and the last step uses $a_\mu^2 \norm{P}_{c, 1} / \norm{P_\mu}_{c, 1} = 1$ and Equation~\eqref{eqn:mu decomp}.
Projecting onto $\ket{0}_{AB}$ gives $P(\bfx)$.
As for the resources, $R_z(x_j)$ is used once in each of the $L_j$ layers with $l \leq L_j$, and $R_z(x_N)$ once in each of the $L_N$ layers.
The gates $E_{j, l}$ are fixed permutations; $\prep$ and the layers $G^{(l)}$ can be realized with multi-controlled $R_y$ and $R_z$ gates, where $\prep$ takes $\ell - 1$ parameters and each $G^{(l)}$ takes $2\ell$, which gives at most $2(L_N + 1)\ell + 2(\ell - 1) = 2(L_N + 2)\ell - 2$ rotation gates in total.
This counts the parameter slots of the template; vanishing angles and inactive branches may reduce the number of gates, and the count excludes the fixed permutations, the decomposition of the multi-controlled rotations into elementary gates and the precision of the angles.
\end{proof}

The query count of Theorem~\ref{thm:mqsp single qubit} is one per unit of degree in every variable, as in univariate QSP, while the number $\ell$ of branches enters only the rotation count and the size of the register $A$.
All parameters are determined by univariate computations: the amplitudes $a_\mu$ are coefficient norms, and the angles of each branch are obtained by the phase-finding algorithms of univariate QSP~\cite{haah2019product, chao2020finding} applied to $\ell$ polynomials of degree at most $L_N$, so no numerical search over the multivariate structure is needed; the precision and running time of these algorithms are not analyzed here.
For $N = 1$ the register $AB'$ is trivial, $\prep$ and $E_{j, l}$ disappear, and $W$ reduces to the sequence of Corollary~\ref{coro:trig qsp}.
Note that the convention $L_N = \max_j L_j$ is a matter of economy: running the sequence in a variable of smaller degree would increase $\ell$ without changing the query counts.

The range of polynomials that the circuit simulates is slightly larger than $\PP_\bfL$.
The constraint $\norm{P}_{c, 1} \leq 1$ can be relaxed to $\sum_\mu \norm{P_\mu}_\infty \leq 1$ for the branches $P_\mu$ in Equation~\eqref{eqn:mu decomp}, so that the same circuit simulates every polynomial in
\begin{equation}\label{eqn:relaxed class}
    \setcond{P \in \laurent}{\deg(P) \preceq \bfL \textrm{, $P$ has parity $\bfL \bmod 2$, } \sum_\mu \norm{P_\mu}_{\infty} \leq 1}
.\end{equation}
Indeed, for $P \neq 0$, with $\alpha = \sum_\mu \norm{P_\mu}_\infty > 0$ and the sum running over the nonzero branches, take $a_\mu = \sqrt{\norm{P_\mu}_\infty / \alpha}$ and let the layers $G^{(l)}$ realize $\alpha P_\mu / \norm{P_\mu}_\infty$, whose absolute value is at most $\alpha \leq 1$; the computation in the proof then goes through unchanged.
Since $\norm{P}_\infty \leq \sum_\mu \norm{P_\mu}_\infty \leq \norm{P}_{c, 1}$, this condition lies between the supremum-norm condition of univariate QSP, to which it reduces for $N = 1$, and the coefficient-norm condition of $\PP_\bfL$.

The query counts in Theorem~\ref{thm:mqsp single qubit} are also necessary, and the parity constraint of $\PP_\bfL$ is intrinsic to the query model.
Both facts follow from the bookkeeping of exponents already visible in Lemma~\ref{lem:trig qsp}: every query to $R_z(x_j)$ multiplies matrix entries by $e^{\mp ix_j/2}$ and thus moves the exponent of $e^{ix_j/2}$ by one, whatever the other gates of the circuit are.
The following proposition makes this precise.

\begin{proposition}\label{prop:scalar lower bound}
Let $P \in \laurent$ be nonzero.
Suppose that a circuit $W(\bfx)$ accesses each variable $x_j$ only through $q_j$ queries to $R_z(x_j)$, with all other gates independent of $\bfx$, and that $W(\bfx)$ simulates $P$ in the sense of Equation~\eqref{eqn:block convention}.
Then $P$ has parity $(q_1 \bmod 2, \ldots, q_N \bmod 2)$ and $\deg_j(P) \leq q_j$ for every $j$.
In particular, every circuit that simulates a polynomial $P \in \PP_\bfL$ with $\deg(P) = \bfL$ uses at least $L_j$ queries to $R_z(x_j)$ for each $j$, which Theorem~\ref{thm:mqsp single qubit} attains.
\end{proposition}
\begin{proof}
Write $W(\bfx) = W_M \cdots W_1$ as a product of gates, each extended by the identity to the whole register of the circuit.
The entries of a gate independent of $\bfx$ are constants, and the nonzero entries of a query $R_z(x_j)$ are $e^{-ix_j/2}$ and $e^{ix_j/2}$.
Since matrix multiplication forms sums of products of entries, induction on $m$ shows that every entry of $W_m \cdots W_1$ is a Laurent polynomial in $e^{ix_1/2}, \ldots, e^{ix_N/2}$ in which, for every monomial, the exponent of $e^{ix_j/2}$ has absolute value at most the number of queries to $R_z(x_j)$ among $W_1, \ldots, W_m$ and the same parity as that number.
For $m = M$, every diagonal entry of the block in Equation~\eqref{eqn:block convention} is an entry of $W(\bfx)$ and equals $P(\bfx)$ for all $\bfx \in \RR^N$.
The functions $e^{i(\bfn \cdot \bfx)/2}$ for $\bfn \in \ZZ^N$ are orthogonal over $[0, 4\pi]^N$, so the coefficients of a Laurent polynomial in $e^{ix_1/2}, \ldots, e^{ix_N/2}$ are determined by its values.
Hence every $\bfn$ with $C_\bfn \neq 0$ satisfies $\abs{n_j} \leq q_j$ and $n_j \equiv q_j \pmod 2$, which is the claim.
\end{proof}

Proposition~\ref{prop:scalar lower bound} does not assume the structure of Figure~\ref{fig:mqsp single qubit}: the other gates may act on any register and may be controlled by it.
The bound concerns exact simulation for every $\bfx \in \RR^N$ and is attained by Theorem~\ref{thm:mqsp single qubit} for all $P \in \PP_\bfL$ with $\deg(P) = \bfL$.

For a fixed order of the scalar queries, Appendix~\ref{sec:sdp} gives an exact semidefinite characterization of the attainable normalization, together with the associated workspace and order constraints.

\subsection{Commuting unitaries}\label{sec:unitary}

We now let the variables be the eigenphases of mutually commuting unitaries $\bfU = (U_1, \ldots, U_N)$ acting on $S$, accessed through the oracles $C_0(U_j)$ and $C_1(U_j)$ in Equation~\eqref{eqn:oracles}, and simulate $P(\bfU)$.
The mechanism is that of QPP reviewed in Section~\ref{sec:preliminaries}.
Restricted to the subspace spanned by $\ket{0}\ket{\chi_k}$ and $\ket{1}\ket{\chi_k}$ of a control qubit and $S$, the oracles act as
\begin{equation}\label{eqn:oracle phase}
    C_s(U_j) = \bigoplus_{k=0}^{d-1} e^{i(2s - 1) x_{j, k}/2}\, R_z(x_{j, k}), \qquad s \in \set{0, 1}
,\end{equation}
so replacing each query $R_z(x_j)$ in the scalar circuit by one of the two oracles, with the control on $B_j$ and the target on $S$, runs the scalar circuit on every common eigenvector simultaneously, up to a scalar phase per eigenvector that records which oracle was used in each layer.
An unbalanced schedule leaves the half-angle phase $e^{i\Delta_j x_{j, k}/2}$, where $\Delta_j$ is the excess of $C_1$ over $C_0$ queries; multiplied by a scalar polynomial of parity $\Delta_j \bmod 2$, this phase yields integer exponents only, and so it can be used to shift the exponents of the simulated polynomial.
For $\bfL^-, \bfL^+ \in \NN^N$ we consider polynomials $P \in \laurentint$ with $\deg^{-}(P) \preceq \bfL^-$, $\deg^{+}(P) \preceq \bfL^+$ and $\norm{P}_{c, 1} \leq 1$, that is, $P(\bfx) = \sum_\bfn C_\bfn e^{i\bfn \cdot \bfx}$ with $C_\bfn = 0$ unless $-L_j^- \leq n_j \leq L_j^+$ for every $j$, and $\sum_\bfn \abs{C_\bfn} \leq 1$.
Set $\bfL = \bfL^- + \bfL^+$ and $\boldsymbol{\Delta} = \bfL^+ - \bfL^-$.
A \emph{query schedule} is a choice of $\sigma_{j, l} \in \set{0, 1}$ for $1 \leq j \leq N$ and $1 \leq l \leq L_j$ with $\sum_{l=1}^{L_j} \sigma_{j, l} = L_j^+$.
With the same gates $\prep$, $G^{(l)}$ and $E_{j, l}$ as in Section~\ref{sec:scalar}, define the query layer
\begin{equation}\label{eqn:asymmetric query layer}
    D_l^\sigma(\bfU) = \left( \prod_{j:\, l \leq L_j} (E_{j, l})_{AB'}^\dagger\, C_{\sigma_{j, l}}(U_j)_{B_j S}\, (E_{j, l})_{AB'} \right) C_{\sigma_{N, l}}(U_N)_{B_N S}
,\end{equation}
where the product ranges over $1 \leq j \leq N - 1$, and the circuit
\begin{equation}\label{eqn:asymmetric circuit}
    V_\sigma(\bfU; a, \theta, \phi) = \prep^\dagger\, G^{(0)} \left[\prod_{l=1}^{L_N} D_l^\sigma(\bfU)\, G^{(l)}\right] \prep
\end{equation}
uses $L_j^+$ queries to $C_1(U_j)$ and $L_j^-$ queries to $C_0(U_j)$.

\begin{theorem}\label{thm:mqsp asymmetric}
Let $\bfL^-, \bfL^+ \in \NN^N$, let $P \in \laurentint$ with $\deg^{-}(P) \preceq \bfL^-$, $\deg^{+}(P) \preceq \bfL^+$ and $\norm{P}_{c, 1} \leq 1$, let $\bfL = \bfL^- + \bfL^+$ with the variables labelled so that $L_N = \max_j L_j$, and let $\ell = \prod_{j=1}^{N-1}(L_j + 1)$.
For every query schedule $\sigma$ there exist $a \in \RR^{\ell}$ and $\theta, \phi \in \RR^{(L_N + 1) \times \ell}$ such that the circuit $V_\sigma(\bfU; a, \theta, \phi)$ in Equation~\eqref{eqn:asymmetric circuit} satisfies
\begin{equation}
    V_\sigma(\bfU; a, \theta, \phi) = \begin{bmatrix}
        P(\bfU) & \ldots \\
        \ldots & \ldots
    \end{bmatrix}
\end{equation}
for every tuple $\bfU$ of mutually commuting unitaries on $S$.
Moreover, $V_\sigma(\bfU; a, \theta, \phi)$ uses exactly $L_j^+$ queries to $C_1(U_j)$ and $L_j^-$ queries to $C_0(U_j)$ for each $j$, and at most $2(L_N + 2)\ell - 2$ single-parameter rotation gates whose parameters depend on $P$.
\end{theorem}
\begin{proof}
Define the centred polynomial
\begin{equation}\label{eqn:asymmetric centering}
    \widetilde{P}(\bfx) = e^{-i\boldsymbol{\Delta} \cdot \bfx/2} P(\bfx) = \sum_{\bfn} C_\bfn\, e^{i(2\bfn - \boldsymbol{\Delta}) \cdot \bfx/2}
.\end{equation}
For every $\bfn$ with $C_\bfn \neq 0$ and every $j$, the exponent $m_j = 2n_j - \Delta_j$ satisfies $-L_j \leq m_j \leq L_j$, because $-L_j^- \leq n_j \leq L_j^+$, and $m_j \equiv L_j \pmod 2$, because $\Delta_j \equiv L_j \pmod 2$.
The map $\bfn \mapsto 2\bfn - \boldsymbol{\Delta}$ is injective, so $\norm{\widetilde{P}}_{c, 1} = \norm{P}_{c, 1} \leq 1$, and hence $\widetilde{P} \in \PP_\bfL$.
Let $a$, $\theta$, $\phi$ be the parameters of Theorem~\ref{thm:mqsp single qubit} for $\widetilde{P}$.
Choose a common eigenbasis $\set{\ket{\chi_k}}_{k=0}^{d-1}$ of $S$, with $U_j\ket{\chi_k} = e^{ix_{j, k}}\ket{\chi_k}$.
All gates other than the oracles act trivially on $S$, and each oracle preserves these eigenspaces, so $V_\sigma(\bfU; a, \theta, \phi)$ is block diagonal in this basis.
By Equation~\eqref{eqn:oracle phase}, its layer $D_l^\sigma(\bfU)$ acts on the block of $\ket{\chi_k}$ as $\left(\prod_{j:\, l \leq L_j} e^{i(2\sigma_{j, l} - 1) x_{j, k}/2}\right) D_l(\bfx_{:, k})$, where the product includes $j = N$ and $D_l(\bfx)$ is the scalar query layer in Equation~\eqref{eqn:scalar query layer}.
Collecting the phases over the layers and using $\sum_{l=1}^{L_j} (2\sigma_{j, l} - 1) = 2L_j^+ - L_j = \Delta_j$ gives
\begin{equation}\label{eqn:asymmetric spectral lift}
    V_\sigma(\bfU; a, \theta, \phi) = \sum_{k=0}^{d-1} e^{i\boldsymbol{\Delta} \cdot \bfx_{:, k}/2}\, W(\bfx_{:, k}; a, \theta, \phi) \ox \ketbra{\chi_k}{\chi_k}
.\end{equation}
Theorem~\ref{thm:mqsp single qubit} and Equation~\eqref{eqn:asymmetric centering} then give
\begin{equation}
    \left(\bra{0}_{AB} \ox I_S\right) V_\sigma(\bfU; a, \theta, \phi) \left(\ket{0}_{AB} \ox I_S\right)
    = \sum_{k=0}^{d-1} e^{i\boldsymbol{\Delta} \cdot \bfx_{:, k}/2}\, \widetilde{P}(\bfx_{:, k}) \ketbra{\chi_k}{\chi_k}
    = \sum_{k=0}^{d-1} P(\bfx_{:, k}) \ketbra{\chi_k}{\chi_k} = P(\bfU)
,\end{equation}
where the last equality is Equation~\eqref{eqn:unitary transformation}.
The query counts are those of the schedule, and the rotation gates are those of Theorem~\ref{thm:mqsp single qubit}.
\end{proof}

Theorem~\ref{thm:mqsp asymmetric} processes all $d$ eigenphases at a cost independent of $d$, in the same way as QPP does for one unitary, and it reduces to QPP~\cite{wang2023quantum} for $N = 1$.
The centring in Equation~\eqref{eqn:asymmetric centering} is used only to synthesize the angles; the imbalance of the schedule supplies the inverse phase, so the simulated operator contains only integer powers of the unitaries and, unlike the scalar circuit $W$ at odd $L_j$, does not depend on the eigenphase representatives.
For this fixed layered construction, the placement of the choices between $C_0(U_j)$ and $C_1(U_j)$ does not matter: any two schedules with the same counts implement the same unitary $V_\sigma$.
For example, $P(U) = U^7$ is simulated with seven queries to $C_1(U)$ and none to $C_0(U)$, whereas Corollary~\ref{thm:mqsp multi qubit} would require seven queries to each oracle, and the bivariate polynomial $P(U_1, U_2) = \tfrac{1}{2} U_1^3 U_2^{-1} + \tfrac{i}{2} U_1^{-2} U_2^4$ is simulated with $\bfL^- = (2, 1)$ and $\bfL^+ = (3, 4)$, that is, with five queries to each unitary.
As in Section~\ref{sec:scalar}, the normalization can be relaxed to $\sum_\mu \norm{\widetilde{P}_\mu}_\infty \leq 1$ for the branches of the centred polynomial, which for $N = 1$ is the condition $\norm{P}_\infty \leq 1$ of QPP.

The balanced even-degree case has a convenient alternating realization. With the same gates $\prep$, $G^{(l)}$ and $E_{j, l}$ as in Section~\ref{sec:scalar}, define the query layer
\begin{equation}\label{eqn:unitary query layer}
    D_l(\bfU) = \left( \prod_{j:\, l \leq L_j} (E_{j, l})_{AB'}^\dagger\, C_{l \bmod 2}(U_j)_{B_j S}\, (E_{j, l})_{AB'} \right) C_{l \bmod 2}(U_N)_{B_N S}
,\end{equation}
so that the layers with odd $l$ query $C_1(U_j)$ and those with even $l$ query $C_0(U_j)$, and the circuit
\begin{equation}\label{eqn:unitary circuit}
    V(\bfU; a, \theta, \phi) = \prep^\dagger\, G^{(0)} \left[\prod_{l=1}^{L_N} D_l(\bfU)\, G^{(l)}\right] \prep
,\end{equation}
which is drawn in Figure~\ref{fig:mqsp multi qubit}.

\begin{figure}[t]
\centering
\resizebox{\linewidth}{!}{
\includegraphics{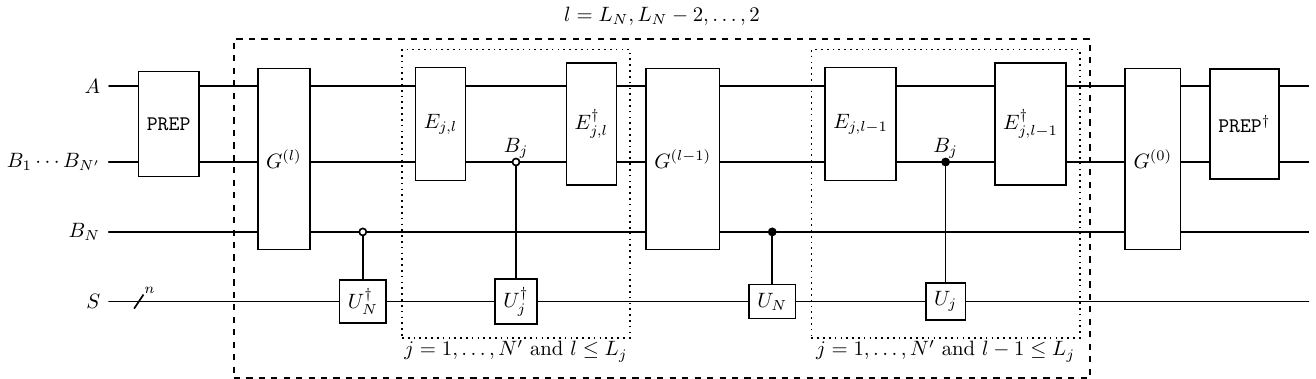}
}
\caption{The circuit $V(\bfU; a, \theta, \phi)$ in Equation~\eqref{eqn:unitary circuit}, where the parameters $a, \theta, \phi$ of $\prep$ and $G^{(l)}$ are the same as in Figure~\ref{fig:mqsp single qubit} and are omitted. Here all entries of $\bfL$ are even, $N' = N - 1$, and $L_N$ is the largest entry of $\bfL$; the layers with even $l$ query $C_0(U_j)$ and those with odd $l$ query $C_1(U_j)$.}
\label{fig:mqsp multi qubit}
\end{figure}

\begin{corollary}\label{thm:mqsp multi qubit}
Let $\bfL \in \NN^N$ with $\bfL \bmod 2 = 0^N$, $P \in \PP_{\bfL, 0}$ and $\ell = \prod_{j=1}^{N-1}(L_j + 1)$.
There exist $a \in \RR^{\ell}$ and $\theta, \phi \in \RR^{(L_N + 1) \times \ell}$ such that the circuit $V(\bfU; a, \theta, \phi)$ in Equation~\eqref{eqn:unitary circuit} satisfies
\begin{equation}
    V(\bfU; a, \theta, \phi) = \begin{bmatrix}
        P(\bfU) & \ldots \\
        \ldots & \ldots
    \end{bmatrix}
\end{equation}
for every tuple $\bfU$ of mutually commuting unitaries on $S$.
Moreover, $V(\bfU; a, \theta, \phi)$ uses exactly $L_j/2$ queries to $C_0(U_j)$ and $L_j/2$ queries to $C_1(U_j)$ for each $j$, and at most $2(L_N + 2)\ell - 2$ single-parameter rotation gates whose parameters depend on $P$.
\end{corollary}
Taking $\bfL^- = \bfL^+ = \bfL/2$ and $\sigma_{j,l} = l \bmod 2$ in Theorem~\ref{thm:mqsp asymmetric} supplies the displayed parameters and establishes the corollary.

Corollary~\ref{thm:mqsp multi qubit} gives the balanced circuit used in the learning model below and retains a query cost independent of the number $d$ of common eigenphases.
The restriction to even degrees comes from this balanced schedule: for odd $L_j$ the uncancelled phase $e^{\pm ix_{j, k}/2}$ depends on the representative of the eigenphase, in agreement with the fact that $P(\bfU)$ is undefined for parities other than $0^N$.

The query counts of Theorem~\ref{thm:mqsp asymmetric} are optimal for each oracle separately when $\bfL^{\pm} = \deg^{\pm}(P)$.
The argument is the directional version of Proposition~\ref{prop:scalar lower bound}: a query to $C_1(U_j)$ can only raise the exponent of $U_j$, and a query to $C_0(U_j)$ can only lower it.

\begin{proposition}\label{prop:directional queries}
Let $P \in \laurentint$ be nonzero.
Suppose that a circuit $V(\bfU)$ accesses each $U_j$ only through $q_j^+$ queries to $C_1(U_j)$ and $q_j^-$ queries to $C_0(U_j)$, with all other gates independent of $\bfU$, and that $V(\bfU)$ simulates $P(\bfU)$ for every tuple $\bfU$ of mutually commuting unitaries on $S$.
Then $q_j^+ \geq \deg_j^{+}(P)$ and $q_j^- \geq \deg_j^{-}(P)$ for every $j$.
\end{proposition}
\begin{proof}
Specialize to $U_j = z_j I_S$ with $\abs{z_j} = 1$, which are mutually commuting unitaries for every $\bfz = (z_1, \ldots, z_N)$ on the torus.
Then $C_1(U_j) = \operatorname{diag}(1, z_j) \ox I_S$ and $C_0(U_j) = \operatorname{diag}(z_j^{-1}, 1) \ox I_S$, and $P(\bfU) = P(\bfz) I_S$ with $P(\bfz) = \sum_\bfn C_\bfn \prod_{j=1}^N z_j^{n_j}$.
As in the proof of Proposition~\ref{prop:scalar lower bound}, induction over the gates shows that every entry of $V(\bfU)$ is a Laurent polynomial in $\bfz$ in which, for every monomial, the exponent of $z_j$ lies in $[-q_j^-, q_j^+]$, since each query to $C_1(U_j)$ raises this exponent by at most one and each query to $C_0(U_j)$ lowers it by at most one.
The diagonal entries of the simulated block equal $P(\bfz)$ for all $\bfz$ on the torus, and the coefficients of a Laurent polynomial on the torus are determined by its values, so every $\bfn$ with $C_\bfn \neq 0$ satisfies $-q_j^- \leq n_j \leq q_j^+$, which is the claim.
\end{proof}

\begin{corollary}\label{coro:optimal asymmetric queries}
Let $P \in \laurentint$ be nonzero with $\norm{P}_{c, 1} \leq 1$.
The minimum numbers of queries to $C_1(U_j)$ and to $C_0(U_j)$ in a circuit that simulates $P(\bfU)$ for all tuples of mutually commuting unitaries are $\deg_j^{+}(P)$ and $\deg_j^{-}(P)$, respectively, and Theorem~\ref{thm:mqsp asymmetric} with $\bfL^{\pm} = \deg^{\pm}(P)$ attains both minima simultaneously for every $j$.
In particular, the minimum total number of queries to $U_j$ and $U_j^\dagger$ is $\deg_j^{+}(P) + \deg_j^{-}(P)$.
\end{corollary}
\begin{proof}
Proposition~\ref{prop:directional queries} gives the lower bounds, and $P$ satisfies the hypotheses of Theorem~\ref{thm:mqsp asymmetric} with $\bfL^{\pm} = \deg^{\pm}(P)$ by the definition of the directional degrees in Equation~\eqref{eqn:directional degree}.
\end{proof}

In particular, Corollary~\ref{thm:mqsp multi qubit} is optimal for every $P \in \PP_{\bfL, 0}$ whose exponents attain both $L_j$ and $-L_j$ in each variable.
As in the scalar case, the bounds concern exact simulation uniformly over all commuting tuples.

\subsection{Comparison with previous work}\label{sec:comparison}

Multivariate polynomial transformations of scalar signals have been pursued along two lines: single-qubit sequences that interleave the signal rotations of several variables, and multi-qubit circuits that realize every monomial separately and combine the monomials.
These works state their query counts for different polynomial classes and signal conventions, so a comparison requires a common setting.
For a fair comparison, we consider the regime where all degrees are equal, $L_j = L$ for every $j$, and the simulation is exact.
Table~\ref{tab:comparison scalar} compares Theorem~\ref{thm:mqsp single qubit} with the constructions that explicitly claim multivariate polynomial transformations of scalar signals~\cite{rossi2022multivariable, yu2024nonasymptotic, tang2026saqnn}.
The query models of these constructions are not interchangeable: a controlled $R_z(x_j)$ cannot be composed from a fixed number of plain $R_z(x_j)$ gates, so the query complexities count the native primitive of each method.

\begin{table}[t]
\centering
\caption{
Comparison of exact multivariate polynomial transformations of $N$ scalar inputs, with equal degrees $L_j = L$.
``Simulation range'' gives the polynomial class, and ``Encoding position'' indicates how the output is represented in the circuit or its measurement.
``Query model'' specifies the signal gate available to the method; ``Query complexity'' counts calls to that gate across all inputs.
``Rotation gates'' counts single-parameter rotations whose angles depend on the target polynomial, including controlled rotations.
``Total dimension'' is the dimension of the full circuit Hilbert space, up to rounding to whole qubits.
Query counts use the model listed in each row.
Here $\CC[\bfx]$ denotes polynomials in real variables $x_j \in [-1, 1]$; $\PP_\bfL$ and $\PP_{\bfL, 0}$ are the normalized trigonometric classes in Equations~\eqref{eqn:class scalar} and~\eqref{eqn:class even}; and $\bfn$ is an integer frequency vector.}
\label{tab:comparison scalar}
\setlength{\tabcolsep}{1em}
\resizebox{\linewidth}{!}{
\begin{tabular}{llllccc}
    \toprule
    Method & Simulation range & \Lscell{Encoding \\ position} & Query model & \Cscell{Query \\ complexity} & \Cscell{Rotation \\ gates} & \Cscell{Total \\ dimension} \\
    \midrule
    \addlinespace
    M-QSP~\cite{rossi2022multivariable} & \Lscell{Restricted subset of \\ $\CC[\bfx]$ with degree $\bfL$} & Amplitude & $R_x(2\arccos x_j)$ & $NL$ & $NL + 1$ & $2$ \\
    \addlinespace
    LCU-based~\cite{yu2024nonasymptotic} & $\PP_\bfL$ & Amplitude & \Lscell{Multi-controlled \\ $R_z(x_j)$} & $\BigTheta{NL(L + 1)^N}$ & $\BigTheta{NL(L + 1)^N}$ & $2^N (L + 1)^N$ \\
    \addlinespace
    SAQNN~\cite{tang2026saqnn} & \Lscell{$P \in \PP_{\bfL, 0}$ \\ ($L$ even)} & \Lscell{Amplitude, read \\ out as probability} & \Lscell{Multi-controlled \\ $R_z(-\bfn \cdot \bfx)$} & $(L + 1)^N$ & $\BigTheta{(L + 1)^N}$ & $(L + 1)^N$ \\
    \addlinespace
    \Lscell{Theorem~\ref{thm:mqsp single qubit} \\ (Ours)} & $\PP_\bfL$ & Amplitude & $R_z(x_j)$ & $NL$ & $\leq 2(L + 2)(L + 1)^{N-1} - 2$ & $2(L + 1)^{N-1}$ \\
    \bottomrule
\end{tabular}
}
\end{table}

M-QSP~\cite{rossi2022multivariable} interleaves the signal rotations of all variables on a single qubit and spends one query per unit of degree, but only for a subclass of polynomials whose characterization is open~\cite{mori2023comment, nemeth2023variants, ito2026polynomial}.
The constructions of Ref.~\cite{yu2024nonasymptotic} and Ref.~\cite{tang2026saqnn} implement every polynomial in $\PP_\bfL$ as an amplitude, the latter for even degrees and with the probability $\abs{P(\bfx)}^2$ as its measured output, by preparing one branch per monomial on an index register for every variable and combining the branches by LCU, so every monomial pays for its own signal rotations.
The SAQNN row uses the compressed encoding with one rotation per frequency; its approximation theorem matches the square root of the measured probability, after a classical rescaling, to targets with values in $[0, 1]$.
If a controlled rotation by $n x_j$ is counted as $\abs{n}$ controlled elementary rotations, both use $\tfrac{1}{2} N [L(L + 2) + (L \bmod 2)] (L + 1)^{N-1}$ of them, whereas Theorem~\ref{thm:mqsp single qubit} uses $NL$ plain signal rotations; this compares the native primitives of the two models rather than a common oracle.

The improvement in query complexity over the branch-per-monomial constructions is therefore exponential in $N$, and the query count matches that of M-QSP while the class is all of $\PP_\bfL$.
To the best of our knowledge, Theorem~\ref{thm:mqsp single qubit} is the first construction that attains the lower bound of Proposition~\ref{prop:scalar lower bound} for every polynomial in $\PP_\bfL$ and every number of variables.
The number of rotation gates is at most a factor $(L + 2)/(L + 1)$ above the number $2(L + 1)^N$ of real coefficients of a polynomial in $\PP_\bfL$, which any fixed circuit family with continuously tunable parameters that implements all of $\PP_\bfL$ has to specify~\cite{casas2023multidimensional}, and the register has dimension about $2(L + 1)^{N-1}$, one factor $L + 1$ smaller than the registers of the branch-per-monomial constructions, because the last variable is processed by QSP rather than stored.
Our construction can itself be read as a linear combination of unitaries over the frequency labels of the first $N - 1$ variables.
A Fourier-series LCU~\cite{childs2017quantum, vanapeldoorn2020quantum} that stores the frequencies of all $N$ variables and builds the controlled powers of each signal from binary powers of one controlled query implements the same $\ell_1$-normalized class with $2^{\ceil{\log_2(L_j + 2)}} - 1$ queries to the signal of $x_j$ in that implementation and one index register per variable, and our construction differs from it by handling the last variable with QSP, which removes that register and relaxes the normalization in $x_N$ to the supremum norm in Equation~\eqref{eqn:relaxed class}.
For $N = 2$, the $U(N)$-QSP construction of Ref.~\cite{lu2026quantum} also applies to scalar signals with $L_j$ queries per variable; we discuss it together with the unitary case below.

For commuting operators, the literature is organized by the input model: explicitly given diagonal operators, block encodings, and controlled unitaries.
Table~\ref{tab:comparison unitary} compares Theorem~\ref{thm:mqsp asymmetric} at $\bfL^- = \floor{\bfL/2}$ and $\bfL^+ = \ceil{\bfL/2}$, so that each operator is queried $L$ times in total, with the constructions that explicitly claim multivariate polynomial transformations of commuting operators~\cite{rosenkranz2025quantum, rossi2022multivariable, bornsweil2023quantum}, again in the regime of equal degrees.
These input models are not interchangeable, so each row retains the native access model of its source.

\begin{table}[t]
\centering
\caption{
Comparison of exact multivariate polynomial transformations of $N$ commuting operators on an $n$-qubit system, with equal degrees $L_j = L$.
``Simulation range'' gives the polynomial class and normalization; ``Query model'' specifies access to the input operators; and ``Query complexity'' counts calls in that model across all inputs.
``Ancilla dimension'' is the total circuit dimension divided by the system dimension $2^n$, up to rounding to whole qubits.
Query counts use the model listed in each row.
Here $\laurentint$ denotes Laurent polynomials in $e^{ix_j}$ for unitary eigenphases $x_j$, while $\CC[\bfx]$ denotes polynomials in the Hermitian eigenvalues $x_j \in [-1, 1]$.
The norms $\norm{P}_{c, 1}$ and $\norm{P}_\infty$ are the coefficient $1$-norm and supremum norm, respectively; $\deg^{+}(P)$ and $\deg^{-}(P)$ give the largest positive and negative exponent magnitudes in each variable, with $\preceq$ denoting componentwise comparison.
The symbol `$\backslash$' indicates a quantity not analyzed in the source.}
\label{tab:comparison unitary}
\setlength{\tabcolsep}{1em}
\resizebox{\linewidth}{!}{
\begin{tabular}{lllcc}
    \toprule
    Method & Simulation range & Query model & \Cscell{Query \\ complexity} & \Cscell{Ancilla \\ dimension} \\
    \midrule
    \addlinespace
    LCU-based~\cite{rosenkranz2025quantum} & \Lscell{$P \in \laurentint$ with exponents of $U_j$ in \\ $L + 1$ consecutive integers and $\norm{P}_{c, 1} \leq 1$} & \Lscell{Explicitly given \\ diagonal unitaries} & $\backslash$ & $(L + 1)^N$ \\
    \addlinespace
    M-QSVT~\cite{rossi2022multivariable} & Same as M-QSP & \Lscell{Block encodings with \\ shared singular vectors} & $NL$ & $\backslash$ \\
    \addlinespace
    MQET~\cite{bornsweil2023quantum} & \Lscell{$P \in \CC[\bfx]$ with $\deg(P) \preceq \bfL$ \\ and $\norm{P}_\infty \leq 1$} & \Lscell{Block encodings of \\ commuting Hermitian matrices} & $\BigO{N(L + 1)^N}$ & $\BigO{L^{N-1}}$ \\
    \addlinespace
    \Lscell{Theorem~\ref{thm:mqsp asymmetric} \\ (Ours)} & \Lscell{$P \in \laurentint$ with $\deg^{-}(P) \preceq \floor{\bfL/2}$, \\ $\deg^{+}(P) \preceq \ceil{\bfL/2}$ and $\norm{P}_{c, 1} \leq 1$} & \Lscell{Controlled $U_j$ and \\ controlled $U_j^\dagger$} & $NL$ & $2(L + 1)^{N-1}$ \\
    \bottomrule
\end{tabular}
}
\end{table}

Ref.~\cite{rosenkranz2025quantum} realizes multivariate Fourier series of explicitly given diagonal Hamiltonians on separate registers by a linear combination of their controlled powers; since the operators are not accessed as oracles, it reports gate counts rather than query counts.
Its symmetric exponent range $-d \leq n_j \leq d$ corresponds to $L = 2d$ in Table~\ref{tab:comparison unitary}.
The index register uses $\ceil{\log_2(L + 1)}$ qubits per variable, giving the listed dimension at the exponent range of our row.
M-QSVT~\cite{rossi2022multivariable} lifts M-QSP to block encodings with shared singular vectors and inherits its restricted class, and Ref.~\cite{mori2024efficient} applies the same class to multivariate state preparation with the operators on separate registers.
The multivariable quantum eigenvalue transformation (MQET) of Ref.~\cite{bornsweil2023quantum} implements every polynomial bounded by one on $[-1, 1]^N$ of commuting Hermitian matrices by expanding it in Chebyshev polynomials of the first $N - 1$ matrices, applying QSP in the last matrix to every term and combining the terms by LCU, so the number of queries grows as $(L + 1)^{N-1}$ times the degree, and the output block encoding carries a subnormalization of order $(L + 3)^{N-1}$.

Theorem~\ref{thm:mqsp asymmetric} uses $NL$ queries for every number of variables, an improvement exponential in $N$ over MQET, and the same count as M-QSVT for the whole class of Laurent polynomials with the prescribed exponent ranges; Corollary~\ref{coro:optimal asymmetric queries} shows that no circuit with access to the two controlled oracles can use fewer queries to either oracle.
The inputs are, however, required as controlled unitaries rather than as block encodings of Hermitian matrices, and the normalization is the coefficient norm rather than the supremum norm of MQET.
For two variables, Ref.~\cite{lu2026quantum} writes a bivariate polynomial of two commuting unitaries as $\sum_r p_r(U_1) q_r(U_2)$, realizes the vectors $(p_r)_r$ and $(q_r)_r$ by two $U(N)$-QSP sequences and contracts them, which uses one query per unit of degree in each variable; finding a decomposition with the required norms is left to numerical optimization, and the extension to more than two variables is stated as an open direction.
Theorem~\ref{thm:mqsp asymmetric} attains the same query count for every number of variables with explicit parameters, and reduces to QPP~\cite{wang2023quantum} for $N = 1$.
\section{Gradient variance of the parameterized circuit}\label{sec:trainability}

The circuit of Section~\ref{sec:unitary} defines a learning model when its rotation angles are optimized for a prediction task. Its input is a tuple $\bfU=(U_1,\ldots,U_N)$ of commuting unitaries on a data register $S$ of dimension $d$. Fix the degree vector with $L_N=\max_jL_j$ and a nonempty set $\cI$ of frequency labels from Section~\ref{sec:scalar}. Each label $\mu=(n_1,\ldots,n_{N-1})$ satisfies $|n_j|\leq L_j$ and $n_j\bmod2=L_j\bmod2$.

The selected labels are represented by orthonormal computational-basis states $\ket\mu$ in $AB'$, where $B'=B_1\cdots B_{N-1}$. Their span is the selected label subspace, whose dimension is
\begin{equation}\label{eqn:learning-labeldimension}
 d_{\cI}:=\dim\operatorname{span}\setcond{\ket\mu}{\mu\in\cI}=|\cI|.
\end{equation}
There are $\ell=\prod_{j=1}^{N-1}(L_j+1)$ admissible frequency labels, and the encoding of Section~\ref{sec:scalar} reserves enough space for all of them. Thus $d_{\cI}\leq\ell\leq\dim(AB')$, with $\dim(AB')=|A|\,2^{N-1}$. Selecting every admissible label gives $d_{\cI}=\ell$. Equality $d_{\cI}=\dim(AB')$ holds exactly when the selected labels form a complete basis of $AB'$. Unselected labels and unused computational-basis states can make the full register larger than the selected label subspace.

The real preparation amplitudes $a_\mu$ are fixed, with $\prep_{AB'}\ket0=\sum_{\mu\in\cI}a_\mu\ket\mu$ and $\sum_\mu|a_\mu|^2=1$. The trainable parameters are $\theta_{l,\mu},\phi_{l,\mu}$ in $G^{(l)}$, for $0\leq l\leq L_N$ and $\mu\in\cI$. The label encoding and query schedule remain fixed, and the same parameters are used for all inputs.

Choose a fixed density matrix $\rho_S$ and initialize the ancillas in $\ket0_{AB}$, where $B=B'B_N$. The prediction is the exact expectation of $Z$ on $B_N$,
\begin{equation}\label{eqn:bp-readout}
 f_Z(\bfU)=\trace{
 (I_{AB'}\ox Z_{B_N}\ox I_S)
 V(\ketbra{0}{0}_{AB}\ox\rho_S)V^\dagger}.
\end{equation}
Here $V$ is the circuit of Section~\ref{sec:unitary}, with its dependence on $\bfU,a,\theta,\phi$ understood. The other registers are traced out, so $f_Z(\bfU)\in[-1,1]$ is an unconditional local readout. Training adjusts the angles to fit these predictions to targets.

The common eigenbasis of the inputs reduces the prediction and its derivatives to single-qubit QSP sequences. Write $U_j\ket{\chi_k}=e^{ix_{j,k}}\ket{\chi_k}$ and $\pi_k=\bra{\chi_k}\rho_S\ket{\chi_k}$. For branch $\mu$, define
\begin{equation}\label{eqn:bp-branch}
 Q_\mu(x)=R_y(\theta_{0,\mu})R_z(\phi_{0,\mu})
 \prod_{l=1}^{L_N}\bigl[R_z(x)R_y(\theta_{l,\mu})R_z(\phi_{l,\mu})\bigr],
\end{equation}
and let $r_\mu(x)=\bra0Q_\mu(x)^\dagger ZQ_\mu(x)\ket0$. On a data eigenspace, each controlled query reduces to $R_z(x)$ up to a scalar phase. Proposition~\ref{prop:bp-reduction} gives
\begin{align}
 f_Z(\bfU)&=\sum_{\mu\in\cI}|a_\mu|^2\sum_k\pi_k r_\mu(x_{N,k}),\label{eqn:learning-model}\\
 \partial_{\alpha_{l,\mu}}f_Z
 &=|a_\mu|^2\sum_k\pi_k\partial_{\alpha_{l,\mu}}r_\mu(x_{N,k}),\label{eqn:learning-gradient}\\
 \partial_{\alpha_{l,\mu}}r_\mu(x)
 &=2\operatorname{Re}\bra0Q_\mu(x)^\dagger Z
       \bigl(\partial_{\alpha_{l,\mu}}Q_\mu(x)\bigr)\ket0.\label{eqn:learning-branchgradient}
\end{align}
Here $\alpha_{l,\mu}$ denotes one angle, with $\alpha\in\{\theta,\phi\}$, and the derivative holds the input and all other parameters fixed. In Equation~\eqref{eqn:bp-branch}, $\partial_{\alpha_{l,\mu}}Q_\mu$ replaces the selected gate by its derivative,
\begin{equation}\label{eqn:learning-gatederivatives}
 \partial_\theta R_y(\theta)=-\frac{i}{2}YR_y(\theta),
 \qquad
 \partial_\phi R_z(\phi)=-\frac{i}{2}ZR_z(\phi).
\end{equation}
The rightmost $R_z(\phi_{L_N,\mu})$ acts on $\ket0$ through a phase, so $\partial_{\phi_{L_N,\mu}}f_Z=0$. The remaining coordinates have the following variance lower bound.

\begin{theorem}\label{thm:bp-bound}
Consider the learning model above with $L_N\geq2$, a fixed commuting input $\bfU$, and a fixed density matrix $\rho_S$ independent of the angles. Initialize all $\theta_{l,\mu},\phi_{l,\mu}$ independently and uniformly on $[0,2\pi)$. Expectation and variance below average over this full joint distribution, including the differentiated coordinate.
For every $\mu\in\cI$ with $a_\mu\neq0$, and for $\alpha=\theta$, $0\leq l\leq L_N$, or $\alpha=\phi$, $0\leq l<L_N$,
\begin{equation}\label{eqn:bp-mainbound}
 \mathbb E_{\theta,\phi}[\partial_{\alpha_{l,\mu}}f_Z]=0,
 \qquad
 \operatorname{Var}_{\theta,\phi}(\partial_{\alpha_{l,\mu}}f_Z)
 \geq\frac{|a_\mu|^4}{64\sqrt{2L_N}}.
\end{equation}
The lower bound holds for every data dimension, query spectrum and density matrix satisfying these assumptions.
\end{theorem}

The proof in Appendix~\ref{sec:bp-covariance} averages the circuit gates directly and retains the cross terms between distinct eigenphases. Proposition~\ref{prop:bp-sharpness} shows that the exponent $L_N^{-1/2}$ is optimal for a bound uniform in the input spectrum and data state. Its numerical constant need not be optimal.

Uniform preparation gives every selected label the same probability, $|a_\mu|^2=1/d_{\cI}$. Choosing equal positive amplitudes realizes it as the pure coherent superposition
\begin{equation}\label{eqn:learning-uniformprep}
 \prep_{AB'}\ket0=\frac1{\sqrt{d_{\cI}}}\sum_{\mu\in\cI}\ket\mu,
 \qquad |a_\mu|^2=\frac1{d_{\cI}}\quad(\mu\in\cI).
\end{equation}
All computational-basis states outside $\cI$ have zero probability. The preparation state remains fixed while the rotation angles are sampled according to Theorem~\ref{thm:bp-bound}. This choice gives the following special case.

\begin{corollary}[Uniform preparation]\label{cor:bp-uniform}
Under the assumptions of Theorem~\ref{thm:bp-bound}, suppose $|a_\mu|^2=1/d_{\cI}$ for every $\mu\in\cI$. Then every parameter coordinate covered by that theorem satisfies
\begin{equation}\label{eqn:bp-uniformbound}
 \operatorname{Var}_{\theta,\phi}(\partial_{\alpha_{l,\mu}}f_Z)
 \geq\frac1{64d_{\cI}^2\sqrt{2L_N}}.
\end{equation}
\end{corollary}
\begin{proof}
Uniform preparation gives $|a_\mu|^4=1/d_{\cI}^2$. Substitution into Equation~\eqref{eqn:bp-mainbound} proves the claim.
\end{proof}

For an $n$-qubit data register, polynomial bounds on $|a_\mu|^{-2}$ and $L_N+1$ give an inverse-polynomial variance lower bound in $n$. These coordinates therefore avoid a barren plateau defined by exponentially decaying initialization variance in the number of data qubits. Under uniform preparation, polynomial bounds on $d_{\cI}$ and $L_N+1$ suffice.

The readout bound alone does not give a loss-gradient bound, because a loss combines predictions and gradients across samples. To address this distinction, fix a supervised data set $\{(\bfU^{(i)},y_i)\}_{i=1}^m$, where $m\geq1$, each $\bfU^{(i)}$ is a commuting tuple on the same register, and $y_i\in[-1,1]$. Use the same fixed PREP, $\rho_S$ and trainable angles for every sample. The empirical squared loss is
\begin{equation}\label{eqn:learning-loss}
 \mathcal L(\theta,\phi)=\frac1{2m}\sum_{i=1}^m
       \bigl[f_Z(\bfU^{(i)})-y_i\bigr]^2.
\end{equation}
Different input tuples need not have a common eigenbasis. The following result retains their gradient correlations under the shared initialization.

\begin{theorem}[Supervised squared loss]\label{thm:bp-loss}
Under the assumptions of Theorem~\ref{thm:bp-bound}, fix the data set and loss in Equation~\eqref{eqn:learning-loss}, independently of the initialization. For every coordinate covered by that theorem,
\begin{align}
 \mathbb E_{\theta,\phi}[\partial_{\alpha_{l,\mu}}\mathcal L]&=0,\label{eqn:bp-lossmean}\\
 \operatorname{Var}_{\theta,\phi}(\partial_{\alpha_{l,\mu}}\mathcal L)
 &\geq\frac{|a_\mu|^4}{64\sqrt{2(L_N+1)}}
 \left[\left(\frac1m\sum_i y_i\right)^2
       +2\left|\frac1m\sum_i y_i\tr\bigl(\rho_SU_N^{(i)}\bigr)\right|^2\right].
       \label{eqn:bp-lossbound}
\end{align}
\end{theorem}

The proof in Appendix~\ref{sec:bp-supervised} separates two terms in the loss gradient by an exact symmetry of the angle distribution. It does not assume that prediction errors and readout gradients are independent. Along with polynomial depth and inverse-polynomial branch weights, an inverse-polynomial lower bound on the bracket in Equation~\eqref{eqn:bp-lossbound} excludes exponentially small loss-gradient variance in the number of data qubits. A zero mean target is allowed when the targets correlate with the first spectral moment of the inputs. If both displayed quantities vanish, this estimate gives no positive bound, although other contributions to the variance may remain.

For exact gradient descent, Proposition~\ref{prop:bp-descent} proves monotone loss decrease, convergence of the gradient norm to zero, and an explicit expected decrease after the first update under uniform initialization. Convergence to zero training loss follows under the additional Jacobian condition in Corollary~\ref{cor:bp-fitting}. That condition is not implied by the coordinate-variance bounds.

The constant Fourier component supplies the input-independent readout bound, as shown in Equation~\eqref{eqn:bp-constantcomponent}. Proposition~\ref{prop:bp-sharpness} also gives a pure-state input that cancels every nonconstant Fourier component after spectral averaging. Moreover, Equation~\eqref{eqn:learning-model} is independent of $U_1,\ldots,U_{N-1}$ when $U_N$ and $\rho_S$ are fixed. The multivariate amplitude representation established in Section~\ref{sec:unitary} therefore requires a different readout to retain sensitivity to those inputs.
\section{Discussions and outlook}\label{sec:conclusion}

We have constructed quantum circuits that implement multivariate polynomial transformations of scalar signals and of commuting unitaries with the minimum number of queries allowed by the degrees.
The scalar circuit implements every trigonometric polynomial in the coefficient-$\ell_1$ class with prescribed degrees and parity, using exactly $L_j$ queries to the $j$-th signal, by running univariate trigonometric QSP in the last variable coherently over the frequencies of the others, so that all branches share the same queries.
Its parameters follow from univariate QSP phase factors and the polynomial coefficients, without a numerical search over the multivariate structure, and the ancillary register grows only logarithmically with the degrees.
The degree argument of Proposition~\ref{prop:scalar lower bound} shows that no circuit with the same query access can use fewer queries for a polynomial that attains the degrees, so that, for every polynomial in the class and every number of variables, processing several signals costs no more queries per signal than processing a single one.

For commuting unitaries, the same circuit with controlled queries implements Laurent polynomials in the unitaries, and the numbers of queries to controlled-$U_j$ and controlled-$U_j^\dagger$ can be chosen equal to the largest positive and negative exponents of $U_j$ in the polynomial, which Corollary~\ref{coro:optimal asymmetric queries} shows to be optimal for each oracle separately.
For a single variable the two constructions reduce to trigonometric QSP and QPP.
Compared with previous work, the query complexity improves exponentially in the number of variables over the constructions that realize every monomial separately, and it matches that of M-QSP while covering every polynomial in the specified coefficient-$\ell_1$ class.

With fixed preparation, Theorem~\ref{thm:bp-bound} bounds the initialization gradient variance of the local $Z$ readout below by $\abs{a_\mu}^4/[64\sqrt{2L_N}]$ for every active coordinate when $L_N\geq2$.
The angles are initialized independently and uniformly, and the bound holds for every fixed data state, spectrum and dimension.
Proposition~\ref{prop:bp-sharpness} shows that the depth exponent is optimal for a bound uniform over these inputs.
The bound excludes exponentially small initialization variance in the number of data qubits when branch weights are at least inverse polynomial and query depth is at most polynomial.
At fixed equal degrees with uniform preparation over all $(L_N+1)^{N-1}$ labels, the coordinate variances decay exponentially with the number of input unitaries.
The contribution that survives every spectral average comes from the constant Fourier component of the response.

For empirical squared loss with targets in $[-1,1]$, Theorem~\ref{thm:bp-loss} retains sample correlations and bounds the initialization gradient variance through the mean target and the correlation between targets and the first spectral moment of the last input unitary.
Under the branch-weight and depth assumptions above, an inverse-polynomial lower bound on either the squared mean target or the squared modulus of this correlation excludes exponentially small loss-gradient variance.
With a constant step size satisfying Proposition~\ref{prop:bp-descent}, exact gradient descent decreases the loss and its gradient norm tends to zero.
The proposition also bounds the expected decrease after the first update using the initialization variance.
Corollary~\ref{cor:bp-fitting} proves geometric convergence to zero training loss if the prediction Jacobian has full row rank and its smallest singular value stays bounded away from zero throughout the trajectory.
This is an additional sufficient condition, which the initialization variance bounds do not establish.

Several directions follow from this work.
The normalization in Theorem~\ref{thm:mqsp single qubit} is a coefficient norm, whereas univariate QSP only requires a bound on the supremum norm; the condition in Equation~\eqref{eqn:relaxed class} lies between the two, and whether the supremum norm of $P$ alone suffices at the same query count remains open.
The number of parameterized rotation gates grows with the number of frequency branches, and reducing this cost for sparse or low-rank coefficient tensors would preserve the query advantage with fewer gates.
The lower bounds concern exact simulation for all inputs; allowing an approximation error or restricting the inputs to a known subset may reduce the number of queries.
The polynomial amplitudes can be used in multivariate state preparation and in functions of commuting unitaries; extending the construction to noncommuting inputs would broaden its applications to non-Hermitian simulation.
For learning tasks, the local $Z$ readout studied here retains only the dependence on $U_N$ when the data state is fixed, and observables that retain the other inputs, together with an analysis of the resulting loss gradients, are needed to assess multivariate learning with this circuit.

\section*{Acknowledgment}
The authors are listed in alphabetical order.
This work was supported by the National Natural Science Foundation of China (Grant. No.~92576114, 12447107), the Guangdong Provincial Quantum Science Strategic Initiative (Grant No.~GDZX2403008, GDZX2503001), and the Guangdong Provincial Key Lab of Integrated Communication, Sensing and Computation for Ubiquitous Internet of Things (Grant No.~2023B1212010007).

\appendix
\section{Gradient analysis for the local \texorpdfstring{$Z$}{Z} readout}\label{sec:gradientvariance}

We prove Theorems~\ref{thm:bp-bound} and~\ref{thm:bp-loss} for the local $Z$ readout with fixed preparation and independently initialized angles. The spectral reduction identifies the prediction function. A direct calculation with the circuit gates gives its gradient variance, a sharp depth exponent for the uniform lower bound, and a squared-loss gradient bound. We then prove descent and convergence statements with their respective assumptions. All randomness comes from the initialization angles, and the readout is an exact expectation value.

\subsection{The multivariate circuit as a parameterized learning model}\label{sec:bp-model}

We use the multivariate circuit of Section~\ref{sec:unitary} to define a family of prediction functions. A model input is a tuple $\bfU=(U_1,\ldots,U_N)$ of commuting unitaries acting on a data register $S$ of dimension $d$. For each input, these unitaries have a common orthonormal eigenbasis,
\begin{equation}\label{eqn:bp-data}
 U_j=\sum_{k=0}^{d-1}e^{ix_{j,k}}\ketbra{\chi_k}{\chi_k},
 \qquad j=1,\ldots,N.
\end{equation}
The ancillary registers are $A$ and $B=B'B_N$, where $B'=B_1\cdots B_{N-1}$. The register $AB'$ stores branch labels, and $B_N$ carries the single-qubit rotations associated with each branch. Choose the degree vector $\bfL=(L_1,\ldots,L_N)$ with $L_N=\max_jL_j$ and parity $\bfp=\bfL\bmod2$, as in the main construction. The circuit is
\begin{equation}\label{eqn:bp-circuit}
 V(\bfU;a,\theta,\phi)
 =\prep_{AB'}^\dagger G^{(0)}
   \left(\prod_{l=1}^{L_N}D_l(\bfU)G^{(l)}\right)
   \prep_{AB'}.
\end{equation}
Here $G^{(l)}$ is the label-controlled rotation layer defined in Section~\ref{sec:scalar}, and $D_l(\bfU)$ is the query layer defined in Section~\ref{sec:unitary}. These operators act on $AB$ and $ABS$, respectively, while PREP acts on $AB'$. Identities on undisplayed registers are understood.

The branch labels are the frequency indices of the first $N-1$ variables in Section~\ref{sec:scalar}. Fix a nonempty label set
\begin{equation}\label{eqn:bp-labels}
 \cI\subseteq\setcond{\mu=(n_1,\ldots,n_{N-1})\in\ZZ^{N-1}}
 {\abs{n_j}\leq L_j,\ n_j\bmod2=p_j\text{ for }1\leq j<N}.
\end{equation}
There are $\ell=\prod_{j=1}^{N-1}(L_j+1)$ admissible frequency labels. The selected label subspace has dimension $d_{\cI}=|\cI|\leq\ell$, as defined in Equation~\eqref{eqn:learning-labeldimension}. In the polynomial-synthesis construction, $\cI$ consists of the labels with $\norm{P_\mu}_{c,1}>0$ for the prescribed polynomial $P$. For the learning model, the selected branches are fixed before initialization and retained throughout training. We use the ancillary registers and computational-basis encoding of Section~\ref{sec:scalar}, with orthonormal states $\ket\mu$ in $AB'$ representing these labels.

The preparation gate creates a coherent superposition of the selected branches,
\begin{equation}\label{eqn:bp-prep}
 \prep_{AB'}\ket{0}=\sum_{\mu\in\cI}a_\mu\ket\mu,
 \qquad \sum_{\mu\in\cI}|a_\mu|^2=1.
\end{equation}
The amplitudes $a_\mu$ are real, following the manuscript's convention, and are held fixed during training. The optimization variables are the rotation angles $\theta_{l,\mu},\phi_{l,\mu}$ in $G^{(l)}$, for $0\leq l\leq L_N$ and $\mu\in\cI$. The same angles are used for every model input and for every data eigenspace within that input.

The input enters through $D_l(\bfU)$. The degree vector, selected labels, their encoding, the permutations $E_{j,l}$ of Section~\ref{sec:scalar}, and the query schedule are fixed parts of the architecture. The unitaries $\bfU$ vary with the model input and are held constant when differentiating with respect to the rotation angles. This construction allows every parity vector $\bfp=\bfL\bmod2$.

Choose a fixed density matrix $\rho_S\succeq0$ with $\tr\rho_S=1$, independent of the trainable angles. Initialize the circuit in $\ketbra{0}{0}_{AB}\ox\rho_S$ and use the local $Z$ readout in Equation~\eqref{eqn:bp-readout}.
For fixed preparation and data state, each choice of $\theta,\phi$ therefore defines a prediction function $\bfU\mapsto f_Z(\bfU)\in[-1,1]$. Dependence on the angles is implicit in $f_Z$. The registers $AB'$ and $S$ are traced out, so Equation~\eqref{eqn:bp-readout} is an unconditional local expectation.

To study the gradients at initialization, fix an input $\bfU$ and draw all angles $\theta_{l,\mu},\phi_{l,\mu}$, $0\leq l\leq L_N$, $\mu\in\cI$, independently and uniformly from $[0,2\pi)$. Each angle occurs once in its branch and is shared across all data eigenspaces. All expectations and variances below average over this full joint distribution, including the angle with respect to which the derivative is taken. The differentiated parameter index is fixed, and the derivative is evaluated at each sampled configuration. The input, $\rho_S$, PREP, and the circuit architecture are held fixed.

\subsection{Spectral reduction on the local readout}\label{sec:bp-reduction}

The constructions of Theorem~\ref{thm:mqsp single qubit} and Corollary~\ref{thm:mqsp multi qubit} resolve the circuit into branches indexed by the frequency label and the common data eigenvector. We apply this decomposition at arbitrary rotation angles to express the local $Z$ readout as a weighted sum of single-qubit QSP responses. This identifies the function evaluated by the chosen readout and reduces the subsequent gradient-variance calculation to the covariance of one branch at two eigenphases.

For each label $\mu\in\cI$, we use the branch sequence $Q_\mu(x)$ and its response $r_\mu(x)$ defined in Section~\ref{sec:trainability}.
Unitarity gives $r_\mu(x)=2|\bra0Q_\mu(x)\ket0|^2-1$, relating the response directly to the QSP polynomial amplitude. Write $\pi_k=\bra{\chi_k}\rho_S\ket{\chi_k}$ for the spectral probabilities of the data state, so that $\pi_k\geq0$ and $\sum_k\pi_k=1$.

\begin{proposition}\label{prop:bp-reduction}
Given the learning model defined above, the local readout satisfies
\begin{equation}\label{eqn:bp-reduction}
 f_Z(\bfU)=\sum_{\mu\in\cI}|a_\mu|^2\sum_{k=0}^{d-1}\pi_k r_\mu(x_{N,k}).
\end{equation}
Fix a branch $\mu\in\cI$, a layer $0\leq l\leq L_N$, and $\alpha\in\{\theta,\phi\}$. The partial derivative $\partial_{\alpha_{l,\mu}}$ is taken with respect to the single real angle $\theta_{l,\mu}$ when $\alpha=\theta$, or $\phi_{l,\mu}$ when $\alpha=\phi$, with all other parameters and the input held fixed. Then
\begin{equation}\label{eqn:bp-gradient}
 \partial_{\alpha_{l,\mu}} f_Z=|a_\mu|^2\sum_k\pi_k\partial_{\alpha_{l,\mu}} r_\mu(x_{N,k}),
 \qquad \mathbb E_{\theta,\phi}[\partial_{\alpha_{l,\mu}} f_Z]=0.
\end{equation}
\end{proposition}
\begin{proof}
On a common data eigenvector, the query interfaces act on $B_j$ as
\begin{equation}
 C_0(U_j)\big|_k=e^{-ix_{j,k}/2}R_z(x_{j,k}),
 \qquad C_1(U_j)\big|_k=e^{ix_{j,k}/2}R_z(x_{j,k}).
\end{equation}
The query phases for $j$ multiply to $e^{ip_jx_{j,k}/2}$ over its $L_j$ layers. The fixed label permutations of Section~\ref{sec:scalar} contribute $\exp(i\sum_{j<N}\mu_jx_{j,k}/2)$. Consequently, the core between the two preparation gates, restricted to label $\mu$ and data eigenvector $\ket{\chi_k}$, acts on $B_N$ as
\begin{equation}\label{eqn:bp-coreblock}
 \exp\left[\frac{i}{2}\left(\sum_{j<N}\mu_jx_{j,k}
             +\bfp\cdot\bfx_{:,k}\right)\right]Q_\mu(x_{N,k}).
\end{equation}
The observable commutes with PREP, so the final inverse preparation cancels inside the trace. The trace over $AB'$ retains equal label indices, and the trace over $S$ retains equal spectral indices. The scalar phases in Equation~\eqref{eqn:bp-coreblock} cancel in each retained term, proving Equation~\eqref{eqn:bp-reduction}, including for $\rho_S$ that is not diagonal in the common basis.

Only branch $\mu$ depends on $\alpha_{l,\mu}$, which gives the derivative in Equation~\eqref{eqn:bp-gradient}. Increasing the selected angle by $2\pi$ multiplies its rotation gate, and hence $Q_\mu(x)$, by $-1$. This sign cancels in $r_\mu(x)$. For every fixed $x$, with all other angles held fixed, $r_\mu(x)$ is therefore continuously differentiable and $2\pi$-periodic in $\alpha_{l,\mu}$. Hence
\begin{align*}
 \frac{1}{2\pi}\int_0^{2\pi}
       \partial_{\alpha_{l,\mu}}r_\mu(x)\,\mathrm{d}\alpha_{l,\mu}
 =\frac{1}{2\pi}\left(
       \left.r_\mu(x)\right|_{\alpha_{l,\mu}=2\pi}
       -\left.r_\mu(x)\right|_{\alpha_{l,\mu}=0}\right)=0.
\end{align*}
The derivatives are bounded because they are continuous on the compact closure of the angle domain. Independence of the uniform angles and Fubini's theorem therefore give
\begin{align*}
 \mathbb E_{\theta,\phi}[\partial_{\alpha_{l,\mu}}f_Z]
 &=\int_{[0,2\pi)^{2d_{\cI}(L_N+1)}}
       \partial_{\alpha_{l,\mu}}f_Z
       \prod_{\nu\in\cI}\prod_{j=0}^{L_N}
       \frac{\mathrm{d}\theta_{j,\nu}\,\mathrm{d}\phi_{j,\nu}}{(2\pi)^2}\\
 &=|a_\mu|^2\sum_{k=0}^{d-1}\pi_k\,
       \mathbb E_{\theta,\phi}[\partial_{\alpha_{l,\mu}}r_\mu(x_{N,k})]\\
 &=|a_\mu|^2\sum_{k=0}^{d-1}\pi_k\,
       \mathbb E_{\theta,\phi}\left[\frac{1}{2\pi}\int_0^{2\pi}
       \partial_{\alpha_{l,\mu}}r_\mu(x_{N,k})\,
       \mathrm{d}\alpha_{l,\mu}\right]\\
 &=0.
\end{align*}
\end{proof}

For every $\mu\in\cI$, the rightmost gate $R_z(\phi_{L_N,\mu})$ in Equation~\eqref{eqn:bp-branch} acts first on $\ket0$ and contributes only a phase. Therefore $\partial_{\phi_{L_N,\mu}}r_\mu(x)=0$ for all $x$, and hence $\partial_{\phi_{L_N,\mu}}f_Z=0$. Equation~\eqref{eqn:bp-reduction} averages the responses over the preparation probabilities $|a_\mu|^2$ and the spectral probabilities $\pi_k$. For this readout the phases carried by the first $N-1$ inputs cancel, leaving a function of $U_N$ and the fixed data state. Equation~\eqref{eqn:bp-gradient} reduces each model derivative to derivatives of one QSP branch. The same angles determine that branch at every eigenphase, so its gradient variance requires the correlations between different eigenphases, computed next.

\subsection{Gradient variance from the circuit gates}\label{sec:bp-covariance}

Fix a coordinate $\alpha_{l,\mu}$ of the full model. Throughout this calculation, $\mathbb E_{\theta,\phi}$ and $\operatorname{Var}_{\theta,\phi}$ average over all initialization angles, including $\alpha_{l,\mu}$. The input, its eigenphases, $\rho_S$ and PREP are fixed. The raw second moment is
\begin{equation}\label{eqn:bp-secondmoment}
 \mathbb E_{\theta,\phi}\bigl[(\partial_{\alpha_{l,\mu}}f_Z)^2\bigr]
 =\int_{[0,2\pi)^{2(L_N+1)}}
       (\partial_{\alpha_{l,\mu}}f_Z)^2
       \prod_{j=0}^{L_N}
       \frac{\mathrm{d}\theta_{j,\mu}\,\mathrm{d}\phi_{j,\mu}}{(2\pi)^2}.
\end{equation}
Only branch $\mu$ depends on this coordinate, since PREP is fixed and $\partial_{\alpha_{l,\mu}}r_\nu(x)=0$ for $\nu\neq\mu$. Thus cross-branch terms vanish identically in this coordinate's squared gradient, and the other branches' angle integrals equal one. A nonlinear loss can couple the branches through its derivative with respect to $f_Z$. Here we calculate the statistics of $\partial_{\alpha_{l,\mu}}f_Z$ itself. Its variance satisfies
\begin{align}
 \operatorname{Var}_{\theta,\phi}(\partial_{\alpha_{l,\mu}}f_Z)
 &:=\mathbb E_{\theta,\phi}\left[\left(\partial_{\alpha_{l,\mu}}f_Z
       -\mathbb E_{\theta,\phi}[\partial_{\alpha_{l,\mu}}f_Z]\right)^2\right]\nonumber\\
 &=\mathbb E_{\theta,\phi}\bigl[(\partial_{\alpha_{l,\mu}}f_Z)^2\bigr]
       -\bigl(\mathbb E_{\theta,\phi}[\partial_{\alpha_{l,\mu}}f_Z]\bigr)^2
 =\mathbb E_{\theta,\phi}\bigl[(\partial_{\alpha_{l,\mu}}f_Z)^2\bigr],\label{eqn:bp-variance}
\end{align}
where Proposition~\ref{prop:bp-reduction} gives the zero mean. At fixed eigenphases $x,y$, the two branch gradients also have zero mean by the same proposition. Their cross second moment therefore equals their covariance, with randomness only in the shared angles. Squaring Equation~\eqref{eqn:bp-gradient} gives
\begin{equation}\label{eqn:bp-pairvariance}
 \operatorname{Var}_{\theta,\phi}(\partial_{\alpha_{l,\mu}}f_Z)
 =|a_\mu|^4\sum_{k,k'}\pi_k\pi_{k'}
   \mathbb E_{\theta,\phi}\bigl[\partial_{\alpha_{l,\mu}}r_\mu(x_{N,k})
                   \partial_{\alpha_{l,\mu}}r_\mu(x_{N,k'})\bigr].
\end{equation}
The terms with $k\neq k'$ remain because every eigenspace uses the same random parameters. We compute these terms directly from the gates and then bound the variance uniformly over the spectral probabilities.

\begin{proof}[Proof of Theorem~\ref{thm:bp-bound}]
The mean vanishes by Proposition~\ref{prop:bp-reduction}. To compute the second moment, use $\ketbra0{0}=(I+Z)/2$ and $\tr Z=0$ to obtain
\begin{equation}\label{eqn:bp-gategradient}
 \partial_{\alpha_{l,\mu}}r_\mu(x)
 =\frac12\tr\!\left[Z\,\partial_{\alpha_{l,\mu}}
       \bigl(Q_\mu(x)ZQ_\mu(x)^\dagger\bigr)\right].
\end{equation}
The identity $\tr A\,\tr B=\tr(A\ox B)$ then gives
\begin{align}\label{eqn:bp-twocopy}
 &\mathbb E_{\theta,\phi}\bigl[\partial_{\alpha_{l,\mu}}r_\mu(x)
                              \partial_{\alpha_{l,\mu}}r_\mu(y)\bigr]\nonumber\\
 &\quad=\frac14\tr\!\left[(Z\ox Z)\,
    \mathbb E_{\theta,\phi}\left[
    \partial_{\alpha_{l,\mu}}\bigl(Q_\mu(x)ZQ_\mu(x)^\dagger\bigr)
    \ox
    \partial_{\alpha_{l,\mu}}\bigl(Q_\mu(y)ZQ_\mu(y)^\dagger\bigr)
    \right]\right].
\end{align}
This expression contains two copies of the circuit, evaluated at fixed $x,y$ and the same angles. Each copy is differentiated at layer $l$ before taking the joint average.

Consider one layer and suppress its indices on $\theta,\phi$. Direct multiplication of its single-qubit gates gives
\begin{align}
 R_y(\theta)R_z(\phi)X R_z(\phi)^\dagger R_y(\theta)^\dagger
 &=\cos\theta\cos\phi\,X+\sin\phi\,Y-\sin\theta\cos\phi\,Z,\nonumber\\
 R_y(\theta)R_z(\phi)Y R_z(\phi)^\dagger R_y(\theta)^\dagger
 &=-\cos\theta\sin\phi\,X+\cos\phi\,Y+\sin\theta\sin\phi\,Z,\label{eqn:bp-pauligates}\\
 R_y(\theta)R_z(\phi)Z R_z(\phi)^\dagger R_y(\theta)^\dagger
 &=\sin\theta\,X+\cos\theta\,Z.\nonumber
\end{align}
For $\beta=\theta$ or $\phi$, the uniform average over this layer obeys
\begin{align*}
 \frac1{2\pi}\int_0^{2\pi}\sin\beta\,\mathrm{d}\beta
 &=\frac1{2\pi}\int_0^{2\pi}\cos\beta\,\mathrm{d}\beta
 =\frac1{2\pi}\int_0^{2\pi}\sin\beta\cos\beta\,\mathrm{d}\beta=0,\\
 \frac1{2\pi}\int_0^{2\pi}\sin^2\beta\,\mathrm{d}\beta
 &=\frac1{2\pi}\int_0^{2\pi}\cos^2\beta\,\mathrm{d}\beta=\frac12.
\end{align*}
Expand the tensor products of the three images in Equation~\eqref{eqn:bp-pauligates} and apply these integrals. Every mixed Pauli input, such as $X\ox Y$, averages to zero. The surviving images lie in the span of $X\ox X$, $Y\ox Y$ and $Z\ox Z$. In that order, their coefficient columns form the first matrix below. Differentiating each image in both copies before averaging gives the other two matrices,
\begin{equation}\label{eqn:bp-matrices}
 \mathsf T=\begin{pmatrix}1/4&1/4&1/2\\1/2&1/2&0\\1/4&1/4&1/2\end{pmatrix},
 \qquad
 \mathsf T_\theta=\begin{pmatrix}1/4&1/4&1/2\\0&0&0\\1/4&1/4&1/2\end{pmatrix},
 \qquad
 \mathsf T_\phi=\begin{pmatrix}1/4&1/4&0\\1/2&1/2&0\\1/4&1/4&0\end{pmatrix}.
\end{equation}
For example, the image of $Z$ is $\sin\theta X+\cos\theta Z$, whose tensor square averages to $(X\ox X+Z\ox Z)/2$. Its $\theta$ derivative has the same averaged tensor square, and its $\phi$ derivative is zero. These give the third columns in Equation~\eqref{eqn:bp-matrices}. Applying the same integrals to the first two lines of Equation~\eqref{eqn:bp-pauligates} gives the other columns. Mixed inputs also average to zero when both copies are differentiated. Thus these matrices record the coefficients produced by gate multiplication and averaging.

Next consider the pair of query gates $R_z(x)$ and $R_z(y)$. For either query angle $x$,
\begin{align}
 R_z(x)XR_z(x)^\dagger&=\cos x\,X+\sin x\,Y,\nonumber\\
 R_z(x)YR_z(x)^\dagger&=-\sin x\,X+\cos x\,Y,\label{eqn:bp-querypauli}\\
 R_z(x)ZR_z(x)^\dagger&=Z.\nonumber
\end{align}
Suppose the incoming operator is $v_1X\ox X+v_2Y\ox Y+v_3Z\ox Z$. After the query pair, the coefficients of $X\ox X$ and $Y\ox Y$ sum to $(v_1+v_2)\cos(x-y)$. The coefficient of $Z\ox Z$ remains $v_3$. Mixed terms are removed by the next averaged rotation layer. Since each matrix in Equation~\eqref{eqn:bp-matrices} has equal first and second columns, the complete update is
\begin{equation}\label{eqn:bp-pairtwirl}
\begin{gathered}
 v\longmapsto\mathsf T\mathsf D(c)v
 \quad\hbox{or}\quad
 v\longmapsto\mathsf T_\alpha\mathsf D(c)v,\\
 v=(v_1,v_2,v_3)^{\mathsf T},\quad
 c=\cos(x-y),\quad
 \mathsf D(c)=\operatorname{diag}(c,c,1).
\end{gathered}
\end{equation}
The factor $\mathsf D(c)$ describes the queries after the following layer is averaged. It is distinct from the physical query block $D_l(\bfU)$.

Starting from $Z\ox Z$, average the independent layers from right to left in Equation~\eqref{eqn:bp-twocopy}. The coefficient vector initially equals $e_z=(0,0,1)^{\mathsf T}$. At the selected layer, replace $\mathsf T$ by $\mathsf T_\alpha$. The final trace with $Z\ox Z$, including its factor $1/4$, selects the third coefficient. Consequently, define
\begin{equation}\label{eqn:bp-kernel}
 K_{\alpha,l}(c)
 =e_z^{\mathsf T}[\mathsf T\mathsf D(c)]^l
   \mathsf T_\alpha[\mathsf D(c)\mathsf T]^{L_N-l}e_z.
\end{equation}
The direct gate calculation proves, for all $x,y\in\RR$,
\begin{equation}\label{eqn:bp-crosscov}
 \mathbb E_{\theta,\phi}[\partial_{\alpha_{l,\mu}}r_\mu(x)
                         \partial_{\alpha_{l,\mu}}r_\mu(y)]
 =K_{\alpha,l}(\cos(x-y)).
\end{equation}
The two copies share each layer's angles throughout this calculation. In particular,
\begin{equation}\label{eqn:bp-branchvariance}
 \operatorname{Var}_{\theta,\phi}(\partial_{\alpha_{l,\mu}}r_\mu(x))
 =K_{\alpha,l}(1).
\end{equation}
Substitution into Equation~\eqref{eqn:bp-pairvariance} yields the exact second moment and variance of the full readout gradient,
\begin{align}
 \mathbb E_{\theta,\phi}\bigl[(\partial_{\alpha_{l,\mu}}f_Z)^2\bigr]
 &=\operatorname{Var}_{\theta,\phi}(\partial_{\alpha_{l,\mu}}f_Z)\nonumber\\
 &=|a_\mu|^4\sum_{k,k'}\pi_k\pi_{k'}
       K_{\alpha,l}(\cos(x_{N,k}-x_{N,k'})).\label{eqn:bp-phasevariance}
\end{align}

All entries of $\mathsf T$, $\mathsf T_\theta$ and $\mathsf T_\phi$ are nonnegative. There are exactly $L_N$ factors $\mathsf D(c)$ in Equation~\eqref{eqn:bp-kernel}. Hence
\begin{equation}\label{eqn:bp-positivepolynomial}
 K_{\alpha,l}(c)=\sum_{m=0}^{L_N}b_m c^m,\qquad b_m\geq0,
\end{equation}
where $\alpha,l$ and $L_N$ are held fixed. Expanding $\cos^m t=2^{-m}\sum_{j=0}^m\binom mj e^{i(m-2j)t}$ gives
\begin{equation}\label{eqn:bp-fourier}
 K_{\alpha,l}(\cos t)=\sum_{h=-L_N}^{L_N}q_{\alpha,l,h}e^{iht},
 \qquad q_{\alpha,l,h}=q_{\alpha,l,-h}\geq0,
 \qquad \sum_hq_{\alpha,l,h}=K_{\alpha,l}(1),
\end{equation}
with
\begin{equation}\label{eqn:bp-qexplicit}
 q_{\alpha,l,h}
 =\sum_{\substack{|h|\leq m\leq L_N\\m-h\ \mathrm{even}}}
 b_m2^{-m}\binom m{(m-h)/2}.
\end{equation}
Equation~\eqref{eqn:bp-phasevariance} therefore becomes
\begin{align}
 \operatorname{Var}_{\theta,\phi}(\partial_{\alpha_{l,\mu}}f_Z)
 &=|a_\mu|^4\sum_{h=-L_N}^{L_N}q_{\alpha,l,h}
       \left|\sum_k\pi_ke^{ihx_{N,k}}\right|^2\nonumber\\
 &=|a_\mu|^4\sum_{h=-L_N}^{L_N}q_{\alpha,l,h}
       \left|\tr(\rho_SU_N^h)\right|^2
 \geq|a_\mu|^4q_{\alpha,l,0}.\label{eqn:bp-spectralvariance}
\end{align}
The last step uses $\tr\rho_S=1$. This formula retains all cross terms between eigenspaces and isolates a contribution that survives every spectral average.

The coefficient $q_{\alpha,l,0}$ has a direct interpretation. Fourier orthogonality and Equation~\eqref{eqn:bp-crosscov} give
\begin{align}
 q_{\alpha,l,0}
 &=\frac1{(2\pi)^2}\int_0^{2\pi}\!\int_0^{2\pi}
       K_{\alpha,l}(\cos(x-y))\,\mathrm{d}x\,\mathrm{d}y\nonumber\\
 &=\mathbb E_{\theta,\phi}\left[\left(\frac1{2\pi}\int_0^{2\pi}
       \partial_{\alpha_{l,\mu}}r_\mu(x)\,\mathrm{d}x\right)^2\right]\nonumber\\
 &=\operatorname{Var}_{\theta,\phi}\left(\partial_{\alpha_{l,\mu}}
       \left[\frac1{2\pi}\int_0^{2\pi}r_\mu(x)\,\mathrm{d}x\right]\right).
       \label{eqn:bp-constantcomponent}
\end{align}
Bounded continuous derivatives justify interchanging differentiation, expectation and integration. The zero mean from Proposition~\ref{prop:bp-reduction} gives the last equality. Thus $q_{\alpha,l,0}$ is the gradient variance of the constant Fourier coefficient of the branch response. The $x,y$ integrals extract this coefficient and impose no distribution on the input spectrum.

To bound $q_{\alpha,l,0}$, first evaluate $K_{\alpha,l}(1)$. Set $v=\mathsf T^{L_N-l}e_z$. Since $\mathsf T$ is doubly stochastic, $v$ has nonnegative entries summing to one. For any probability vector, $(\mathsf Tv)_3=1/4+v_3/4\leq1/2$, so $v_3\leq1/2$ if $L_N-l\geq1$. For $l\geq1$, every entry of $e_z^{\mathsf T}\mathsf T^l$ lies in $[1/4,1/2]$. Indeed, $e_z^{\mathsf T}\mathsf T=(1/4,1/4,1/2)$, and subsequent multiplication forms convex combinations of these entries. With $\boldsymbol1=(1,1,1)^{\mathsf T}$,
\begin{equation}
 \boldsymbol1^{\mathsf T}\mathsf T_\theta v=\frac{1+v_3}{2},
 \qquad
 \boldsymbol1^{\mathsf T}\mathsf T_\phi v=1-v_3.
\end{equation}
Both quantities belong to $[1/2,1]$ for the coordinates in Theorem~\ref{thm:bp-bound}. Since $\mathsf T_\alpha v$ is nonnegative, $e_z^{\mathsf T}\mathsf T^l\mathsf T_\alpha v$ lies in $[1/8,1/2]$ for $l\geq1$. For $l=0$, Equation~\eqref{eqn:bp-kernel} equals $(1+v_3)/4$ for $\theta$ and $(1-v_3)/4$ for $\phi$. Consequently,
\begin{equation}\label{eqn:bp-scalarbound}
 \frac18\leq K_{\alpha,l}(1)\leq\frac12
 \quad\text{for }\alpha=\theta,\ 0\leq l\leq L_N,
 \quad\text{or }\alpha=\phi,\ 0\leq l<L_N.
\end{equation}
This estimate also holds for $L_N=1$. The excluded coordinate satisfies $\mathsf T_\phi e_z=0$, so $K_{\phi,L_N}(c)=0$ for every $c$.

Equation~\eqref{eqn:bp-qexplicit} expresses the zero coefficient as
\begin{equation}\label{eqn:bp-evenbound}
 q_{\alpha,l,0}
 =\sum_{\substack{0\leq m\leq L_N\\m\ \mathrm{even}}}b_m2^{-m}\binom m{m/2}
 \geq 4^{-\lfloor L_N/2\rfloor}
       \binom{2\lfloor L_N/2\rfloor}{\lfloor L_N/2\rfloor}
       \sum_{m\ \mathrm{even}}b_m.
\end{equation}
Indeed, the normalized central binomial coefficients decrease with the nonnegative integer $r$, since
\begin{equation}\label{eqn:bp-binomialratio}
 \frac{4^{-(r+1)}\binom{2r+2}{r+1}}{4^{-r}\binom{2r}{r}}
 =\frac{2r+1}{2r+2}<1.
\end{equation}
Also,
\begin{equation}\label{eqn:bp-paritysum}
 \sum_{m\ \mathrm{even}}b_m
 =\frac{K_{\alpha,l}(1)+K_{\alpha,l}(-1)}2.
\end{equation}
We must control cancellation between the variance $K_{\alpha,l}(1)$ and the cross second moment $K_{\alpha,l}(-1)$ at eigenphases separated by $\pi$.

For $L_N\geq5$, direct multiplication gives
\begin{equation}\label{eqn:bp-threeblocks}
 64\mathsf T^3=\begin{pmatrix}21&21&22\\22&22&20\\21&21&22\end{pmatrix},
 \qquad
 64[\mathsf T\mathsf D(-1)]^3
 =\begin{pmatrix}-5&-5&6\\-14&-14&4\\-5&-5&6\end{pmatrix}.
\end{equation}
With absolute values and inequalities interpreted entrywise,
\begin{equation}\label{eqn:bp-contraction}
 |[\mathsf T\mathsf D(-1)]^3|\leq\frac34\mathsf T^3,
 \qquad
 |[\mathsf D(-1)\mathsf T]^3|\leq\frac34\mathsf T^3.
\end{equation}
The second inequality follows by conjugating the first matrix product with $\mathsf D(-1)$, which preserves entrywise absolute values. At least one of $l,L_N-l$ is at least three. Isolate three consecutive factors on that side of Equation~\eqref{eqn:bp-kernel} and apply Equation~\eqref{eqn:bp-contraction}. For the remaining factors, use $|\mathsf T\mathsf D(-1)|=|\mathsf D(-1)\mathsf T|=\mathsf T$ and $|\mathsf T_\alpha|=\mathsf T_\alpha$. The entrywise inequality $|AB|\leq|A||B|$ then gives
\begin{equation}
 |K_{\alpha,l}(-1)|\leq\frac34K_{\alpha,l}(1),
 \qquad
 \sum_{m\ \mathrm{even}}b_m
 \geq\frac18K_{\alpha,l}(1)\geq\frac1{64}.
\end{equation}
For $L_N=2,3,4$, direct multiplication in Equation~\eqref{eqn:bp-kernel} gives the following lists, ordered by increasing $l$,
\begin{equation}\label{eqn:bp-smalldepths}
\begin{array}{c|c|c}
 L_N&\left(\dfrac{K_{\theta,l}(1)+K_{\theta,l}(-1)}2\right)_{l=0}^{L_N}
  &\left(\dfrac{K_{\phi,l}(1)+K_{\phi,l}(-1)}2\right)_{l=0}^{L_N-1}\\[2mm]
 \hline
 2&(7,5,7)/32&(3,3)/32\\[1mm]
 3&(11,9,9,11)/64&(2,3,2)/32\\[1mm]
 4&(91,65,65,65,91)/512&(47,55,55,47)/512
\end{array}
\end{equation}
Every entry is at least $1/64$. For every integer $r\geq1$,
\begin{equation}\label{eqn:bp-binomiallower}
 4^{-r}\binom{2r}{r}\geq\frac1{2\sqrt r}.
\end{equation}
Equality holds at $r=1$. The induction step follows from Equation~\eqref{eqn:bp-binomialratio} and
\begin{equation*}
 \left(\frac{2r+1}{2r+2}\right)^2-\frac r{r+1}
 =\frac1{4(r+1)^2}>0.
\end{equation*}
Combining these estimates with Equation~\eqref{eqn:bp-evenbound} proves
\begin{equation}\label{eqn:bp-zeromodelower}
 q_{\alpha,l,0}\geq\frac{4^{-\lfloor L_N/2\rfloor}}{64}
       \binom{2\lfloor L_N/2\rfloor}{\lfloor L_N/2\rfloor}
 \geq\frac1{64\sqrt{2L_N}},\qquad L_N\geq2.
\end{equation}
Finally, Equation~\eqref{eqn:bp-spectralvariance} yields
\begin{equation}
 \operatorname{Var}_{\theta,\phi}(\partial_{\alpha_{l,\mu}}f_Z)
 \geq|a_\mu|^4q_{\alpha,l,0}
 \geq\frac{|a_\mu|^4}{64\sqrt{2L_N}}.
\end{equation}
No step depends on $d$, the query spectrum or the choice of $\rho_S$.
\end{proof}

The exact variance also supplies an upper estimate. Since $|\tr(\rho_SU_N^h)|\leq1$, Equations~\eqref{eqn:bp-spectralvariance} and~\eqref{eqn:bp-scalarbound} imply
\begin{equation}\label{eqn:bp-upperbound}
 \operatorname{Var}_{\theta,\phi}(\partial_{\alpha_{l,\mu}}f_Z)
 \leq|a_\mu|^4\sum_hq_{\alpha,l,h}
 =|a_\mu|^4K_{\alpha,l}(1)\leq\frac{|a_\mu|^4}{2}.
\end{equation}
The next result identifies inputs that attain the smallest variance and proves that the depth exponent in Theorem~\ref{thm:bp-bound} cannot be improved uniformly over inputs.

\begin{proposition}[Sharp depth dependence]\label{prop:bp-sharpness}
Fix $L_N\geq2$, a branch and an active coordinate as in Theorem~\ref{thm:bp-bound}. Among admissible choices of data dimension, commuting input and data state, the minimum and maximum readout-gradient variances are
\begin{equation}\label{eqn:bp-extrema}
 \min\operatorname{Var}_{\theta,\phi}(\partial_{\alpha_{l,\mu}}f_Z)
 =|a_\mu|^4q_{\alpha,l,0},\qquad
 \max\operatorname{Var}_{\theta,\phi}(\partial_{\alpha_{l,\mu}}f_Z)
 =|a_\mu|^4K_{\alpha,l}(1).
\end{equation}
Both extrema can be attained with pure data states. Moreover,
\begin{equation}\label{eqn:bp-zeroasymptotic}
 \frac1{64\sqrt{2L_N}}\leq q_{\alpha,l,0}
 \leq4^{-L_N}\binom{2L_N}{L_N}
 \leq\frac1{\sqrt{L_N+1}}.
\end{equation}
Thus the minimum variance has order $|a_\mu|^4L_N^{-1/2}$, with constants uniform over the active coordinates.
\end{proposition}
\begin{proof}
Equation~\eqref{eqn:bp-spectralvariance} bounds every variance between the two values in Equation~\eqref{eqn:bp-extrema}. A state supported on one eigenspace attains the maximum. For the minimum, choose any power of two $d>L_N$ and set
\begin{equation}\label{eqn:bp-rootinput}
 U_N=\sum_{k=0}^{d-1}e^{2\pi i k/d}\ketbra{k}{k},\qquad
 \rho_S=\ketbra{\psi}{\psi},\qquad
 \ket\psi=\frac1{\sqrt d}\sum_{k=0}^{d-1}\ket k.
\end{equation}
Take $U_j=I_S$ for $j<N$. Then $\tr(\rho_SU_N^h)=0$ for every integer $0<|h|\leq L_N$, while $\tr\rho_S=1$. Equation~\eqref{eqn:bp-spectralvariance} attains its zero-frequency term exactly.

It remains to prove the upper bound on that term. For $u\in[0,1]$, all matrix entries in Equation~\eqref{eqn:bp-kernel} are nonnegative, and $\mathsf T_\alpha\leq\mathsf T$ entrywise. Hence
\begin{equation*}
 K_{\alpha,l}(u)\leq e_z^{\mathsf T}[\mathsf T\mathsf D(u)]^{L_N}\mathsf Te_z.
\end{equation*}
Let $v_0=\mathsf Te_z$ and $v_{j+1}=\mathsf T\mathsf D(u)v_j$. Each $v_j$ is nonnegative and satisfies $(v_j)_3\leq\boldsymbol1^{\mathsf T}v_j/2$. This holds because $\mathsf T$ is column stochastic and every entry in its third row is at most $1/2$. Consequently,
\begin{equation*}
 \boldsymbol1^{\mathsf T}v_{j+1}
 =u\boldsymbol1^{\mathsf T}v_j+(1-u)(v_j)_3
 \leq\frac{1+u}{2}\boldsymbol1^{\mathsf T}v_j,
 \qquad
 K_{\alpha,l}(u)\leq\frac12\left(\frac{1+u}{2}\right)^{L_N}.
\end{equation*}
The coefficients in Equation~\eqref{eqn:bp-positivepolynomial} are nonnegative, so $K_{\alpha,l}(\cos x)\leq K_{\alpha,l}(|\cos x|)$. Fourier orthogonality now yields
\begin{align*}
 q_{\alpha,l,0}
 &\leq\frac1{4\pi}\int_0^{2\pi}
             \left(\frac{1+|\cos x|}{2}\right)^{L_N}\,\mathrm{d}x
 =\frac2\pi\int_0^{\pi/4}\cos^{2L_N}x\,\mathrm{d}x\\
 &\leq\frac2\pi\int_0^{\pi/2}\cos^{2L_N}x\,\mathrm{d}x
 =4^{-L_N}\binom{2L_N}{L_N}.
\end{align*}
The last integral follows by integration by parts, starting from $\int_0^{\pi/2}1\,\mathrm{d}x=\pi/2$. Finally, $4^{-r}\binom{2r}{r}\leq(r+1)^{-1/2}$ follows by induction from equality at $r=0$, Equation~\eqref{eqn:bp-binomialratio}, and
\begin{equation*}
 \frac{r+1}{r+2}-\left(\frac{2r+1}{2r+2}\right)^2
 =\frac{3r+2}{4(r+1)^2(r+2)}\geq0.
\end{equation*}
Together with Equation~\eqref{eqn:bp-zeromodelower}, this proves Equation~\eqref{eqn:bp-zeroasymptotic}.
\end{proof}

Each query conjugation introduces only the harmonics $0,\pm1$, so $r_\mu(x)$ has integer Fourier support in $[-L_N,L_N]$. The input in Equation~\eqref{eqn:bp-rootinput} therefore also satisfies, for every choice of angles,
\begin{equation}\label{eqn:bp-centeredcancellation}
 \frac1d\sum_{k=0}^{d-1}\left[
 r_\mu(2\pi k/d)-\frac1{2\pi}\int_0^{2\pi}r_\mu(x)\,\mathrm{d}x\right]=0.
\end{equation}
Every parameter derivative of this expression vanishes. A uniform positive variance bound for the spectrally averaged nonconstant response is therefore impossible, even for pure data states. This does not say that the nonconstant Fourier coefficients themselves have zero gradients.

Spectral concentration gives a stronger lower bound, because the nonzero spectral moments also remain large.

\begin{proposition}[Spectrally concentrated inputs]\label{prop:bp-concentrated}
Suppose $L_N\geq1$ and eigenphase representatives can be chosen such that, for some $x_0\in\RR$ and $0\leq\Delta\leq\pi/(4L_N)$,
\begin{equation}
 \pi_k>0\ \Longrightarrow\ |x_{N,k}-x_0|\leq\Delta.
\end{equation}
For every $\mu\in\cI$ and every coordinate in the ranges of Equation~\eqref{eqn:bp-scalarbound},
\begin{equation}
 \operatorname{Var}_{\theta,\phi}(\partial_{\alpha_{l,\mu}}f_Z)\geq\frac{|a_\mu|^4}{16}.
\end{equation}
If $\rho_S$ is supported on one eigenspace of $U_N$, the lower bound is $|a_\mu|^4/8$.
\end{proposition}
\begin{proof}
For every integer $|h|\leq L_N$,
\begin{equation}
 \operatorname{Re}\left[e^{-ihx_0}\tr(\rho_SU_N^h)\right]
 =\sum_k\pi_k\cos(h(x_{N,k}-x_0))
 \geq\cos(|h|\Delta)\geq\frac1{\sqrt2}.
\end{equation}
Thus $|\tr(\rho_SU_N^h)|^2\geq1/2$ on the Fourier support. Equations~\eqref{eqn:bp-spectralvariance} and~\eqref{eqn:bp-scalarbound} give
\begin{equation}
 \operatorname{Var}_{\theta,\phi}(\partial_{\alpha_{l,\mu}}f_Z)
 \geq\frac{|a_\mu|^4}{2}\sum_hq_{\alpha,l,h}
 =\frac{|a_\mu|^4}{2}K_{\alpha,l}(1)
 \geq\frac{|a_\mu|^4}{16}.
\end{equation}
On a single eigenspace, all spectral moments have modulus one, giving $|a_\mu|^4K_{\alpha,l}(1)\geq|a_\mu|^4/8$. Both arguments apply to pure and mixed states.
\end{proof}

\begin{remark}[The shallow case]\label{rem:bp-shallow}
At $L_N=1$, direct multiplication gives
\begin{equation}
 r_\mu(x)=\cos\theta_{0,\mu}\cos\theta_{1,\mu}
 -\sin\theta_{0,\mu}\sin\theta_{1,\mu}\cos(\phi_{0,\mu}+x).
\end{equation}
Consequently,
\begin{equation}
 \operatorname{Var}_{\theta,\phi}(\partial_{\phi_{0,\mu}}f_Z)
 =\frac{|a_\mu|^4}{8}|\tr(\rho_SU_N)|^2.
\end{equation}
This can vanish for a pure input. For example, $U_N=\operatorname{diag}(1,-1)$ and $\ket\psi=(\ket0+\ket1)/\sqrt2$ give $\bra\psi U_N\ket\psi=0$. Purity alone therefore supplies no uniform lower bound at this depth. Theorem~\ref{thm:bp-bound} applies to every integer $L_N\geq2$, with either parity.
\end{remark}

\begin{remark}[Branch weights and the scope of the conclusion]\label{rem:bp-scope}
Theorem~\ref{thm:bp-bound} applies to every fixed $\rho_S$, including $I_S/d$ and states that do not commute with $U_N$. For $d=2^n$, suppose $|a_\mu|^2\geq c_1n^{-a}$ and $L_N+1\leq c_2n^b$ for all sufficiently large $n$. Here $c_1,c_2>0$ and $a,b\geq0$ are constants independent of $n$. Then
\begin{equation}
 \operatorname{Var}_{\theta,\phi}(\partial_{\alpha_{l,\mu}}f_Z)
 \geq\frac{c_1^2}{64\sqrt{2c_2}}n^{-(2a+b/2)}.
\end{equation}
The variance lower bound is therefore inverse polynomial in the number of data qubits under these scaling assumptions.

The branch weights can introduce a different dependence on system size. With equal degrees $L_j=L_N\geq2$, all labels selected and uniform preparation, $d_{\cI}=\ell=(L_N+1)^{N-1}$. Corollary~\ref{cor:bp-uniform} and Equation~\eqref{eqn:bp-upperbound} then give
\begin{equation}\label{eqn:bp-branchscaling}
 \frac1{64\sqrt{2L_N}(L_N+1)^{2N-2}}
 \leq\operatorname{Var}_{\theta,\phi}(\partial_{\alpha_{l,\mu}}f_Z)
 \leq\frac1{2(L_N+1)^{2N-2}}.
\end{equation}
At fixed $L_N$, both bounds decay exponentially with the number $N$ of input unitaries.

Equation~\eqref{eqn:bp-constantcomponent} identifies the contribution responsible for the lower bound uniform in the data state as the gradient of the constant Fourier coefficient. The present calculation assumes fixed PREP, noiseless gates and exact readout. It concerns individual coordinate gradients at initialization. Appendix~\ref{sec:bp-supervised} treats the empirical squared loss and exact gradient descent under these assumptions. Other losses, arbitrary parameter directions and robustness to noise require further analysis. Equation~\eqref{eqn:bp-reduction} also shows that this local readout is independent of $U_1,\ldots,U_{N-1}$ when $U_N$ and $\rho_S$ are fixed.
\end{remark}

\subsection{Supervised squared loss and gradient descent}\label{sec:bp-supervised}

We now consider the fixed data set and empirical squared loss in Equation~\eqref{eqn:learning-loss}. All samples share the same angles, PREP and data state. In particular, the prediction errors are functions of the same random angles as the readout gradients. For a fixed coordinate,
\begin{equation}\label{eqn:bp-lossgradient}
 \partial_{\alpha_{l,\mu}}\mathcal L
 =\frac1m\sum_i\bigl[f_Z(\bfU^{(i)})-y_i\bigr]
          \partial_{\alpha_{l,\mu}}f_Z(\bfU^{(i)}).
\end{equation}
The calculation below keeps both the sample cross terms and the dependence between the two factors in each summand.

\begin{proof}[Proof of Theorem~\ref{thm:bp-loss}]
The loss is continuously differentiable and $2\pi$-periodic in each angle. Holding all other angles fixed gives
\begin{equation*}
 \frac1{2\pi}\int_0^{2\pi}\partial_{\alpha_{l,\mu}}\mathcal L
       \,\mathrm{d}\alpha_{l,\mu}
 =\frac{\left.\mathcal L\right|_{\alpha_{l,\mu}=2\pi}
       -\left.\mathcal L\right|_{\alpha_{l,\mu}=0}}{2\pi}=0.
\end{equation*}
Boundedness on the angle torus and Fubini's theorem prove Equation~\eqref{eqn:bp-lossmean}. Thus the loss-gradient variance equals its raw second moment.

Consider the simultaneous shifts $\theta_{0,\nu}\mapsto\theta_{0,\nu}+\pi$ modulo $2\pi$ for every $\nu\in\cI$. They preserve the joint initialization measure. Since $R_y(\pi)^\dagger ZR_y(\pi)=-Z$, they send every prediction $f_Z(\bfU^{(i)})$ to its negative. Differentiating this identity shows that every readout gradient also changes sign. The product of a prediction and its gradient is therefore unchanged, whereas a target times a gradient changes sign. Measure invariance gives
\begin{equation}\label{eqn:bp-losssymmetry}
 \mathbb E_{\theta,\phi}\left[
 \left(\frac1m\sum_i f_Z(\bfU^{(i)})
                    \partial_{\alpha_{l,\mu}}f_Z(\bfU^{(i)})\right)
 \left(\frac1m\sum_j y_j
                    \partial_{\alpha_{l,\mu}}f_Z(\bfU^{(j)})\right)\right]=0.
\end{equation}
This cancellation follows from the invariance of the initialization measure.

For two samples, apply Equation~\eqref{eqn:bp-crosscov} to their respective spectral decompositions. The two eigenbases can differ, because the readout reduction holds separately for each input. Expanding Equation~\eqref{eqn:bp-fourier} in the resulting double spectral sum gives
\begin{align}\label{eqn:bp-samplecross}
 &\mathbb E_{\theta,\phi}\left[
       \partial_{\alpha_{l,\mu}}f_Z(\bfU^{(i)})
       \partial_{\alpha_{l,\mu}}f_Z(\bfU^{(j)})\right]\nonumber\\
 &\quad=|a_\mu|^4\sum_{h=-L_N}^{L_N}q_{\alpha,l,h}
       \tr\bigl(\rho_S(U_N^{(i)})^h\bigr)
       \overline{\tr\bigl(\rho_S(U_N^{(j)})^h\bigr)}.
\end{align}
Although individual terms can be complex, the $h$ and $-h$ terms are conjugates, so the sum is real. Squaring Equation~\eqref{eqn:bp-lossgradient}, using Equation~\eqref{eqn:bp-losssymmetry} and summing Equation~\eqref{eqn:bp-samplecross} with the real target weights yields the exact identity
\begin{align}\label{eqn:bp-lossvarianceexact}
 \operatorname{Var}_{\theta,\phi}(\partial_{\alpha_{l,\mu}}\mathcal L)
 &=\mathbb E_{\theta,\phi}\left[
       \left(\frac1m\sum_i f_Z(\bfU^{(i)})
              \partial_{\alpha_{l,\mu}}f_Z(\bfU^{(i)})\right)^2\right]\nonumber\\
 &\quad+|a_\mu|^4\sum_{h=-L_N}^{L_N}q_{\alpha,l,h}
       \left|\frac1m\sum_i y_i
              \tr\bigl(\rho_S(U_N^{(i)})^h\bigr)\right|^2.
\end{align}
The first term includes the dependence on all branches through the predictions. Both terms are nonnegative.

To obtain a bound that can also apply to targets of zero mean, we retain $h=\pm1$ as well as $h=0$. The coefficients of odd powers in Equation~\eqref{eqn:bp-positivepolynomial} satisfy
\begin{equation}\label{eqn:bp-oddmass}
 \sum_{j\ \mathrm{odd}}b_j
 =\frac{K_{\alpha,l}(1)-K_{\alpha,l}(-1)}2\geq\frac1{64},
 \qquad L_N\geq2.
\end{equation}
For $L_N\geq5$, this follows from the same contraction as Equation~\eqref{eqn:bp-contraction}, since $|K_{\alpha,l}(-1)|\leq3K_{\alpha,l}(1)/4$ and $K_{\alpha,l}(1)\geq1/8$. For $L_N=2,3,4$, direct evaluation of Equation~\eqref{eqn:bp-kernel} gives, respectively, the minimum odd-power sums over all active coordinates as $1/16$, $13/128$ and $19/256$. Each exceeds $1/64$.

For an odd exponent $j=2r+1$, its contribution to the first harmonic is
\begin{equation*}
 2^{-j}\binom j{(j-1)/2}
 =4^{-(r+1)}\binom{2r+2}{r+1}.
\end{equation*}
Using $r+1\leq\lceil L_N/2\rceil$, Equations~\eqref{eqn:bp-binomialratio}, \eqref{eqn:bp-binomiallower} and~\eqref{eqn:bp-oddmass} give
\begin{equation}\label{eqn:bp-firstharmonic}
 q_{\alpha,l,1}=q_{\alpha,l,-1}
 \geq\frac{4^{-\lceil L_N/2\rceil}}{64}
       \binom{2\lceil L_N/2\rceil}{\lceil L_N/2\rceil}
 \geq\frac1{64\sqrt{2(L_N+1)}}.
\end{equation}
Equation~\eqref{eqn:bp-zeromodelower} gives the same last lower bound for $q_{\alpha,l,0}$. Keeping only $h=0,\pm1$ in the second term of Equation~\eqref{eqn:bp-lossvarianceexact} proves Equation~\eqref{eqn:bp-lossbound}.
\end{proof}

A positive bracket in Equation~\eqref{eqn:bp-lossbound} supplies a sufficient condition for positive loss-gradient variance. For example, two targets $y_1=1$, $y_2=-1$ and inputs $U_N^{(1)}=iI_S$, $U_N^{(2)}=-iI_S$ have zero mean target and satisfy
\begin{equation*}
 \frac12\sum_{i=1}^2y_i\tr(\rho_SU_N^{(i)})=i.
\end{equation*}
Their loss-gradient bound is therefore positive. Conversely, when both quantities in the bracket vanish, higher spectral moments and the first term of Equation~\eqref{eqn:bp-lossvarianceexact} can still contribute. None of these initialization estimates can be reapplied to later iterates without analyzing the changed angle distribution.

We next give an optimization statement for exact gradients. Let $\nabla\mathcal L$ collect the $d_{\cI}(2L_N+1)$ coordinates $\theta_{l,\mu}$ with $0\leq l\leq L_N$ and $\phi_{l,\mu}$ with $0\leq l<L_N$. The unused angles $\phi_{L_N,\mu}$ remain fixed. Coordinates on a branch with $a_\mu=0$ contribute zero and may be retained in this vector. Regard the listed angles as a real vector in these formulas, using periodicity if a representation in $[0,2\pi)$ is desired.

\begin{proposition}[Gradient descent on the squared loss]\label{prop:bp-descent}
For the fixed model and data set in Theorem~\ref{thm:bp-loss}, choose a constant step size
\begin{equation}\label{eqn:bp-stepsize}
 0<\eta\leq\frac1{3(2L_N+1)\max_{\mu\in\cI}|a_\mu|^2}.
\end{equation}
Starting from any angles, update the listed coordinates by
\begin{equation}\label{eqn:bp-update}
 (\theta^{(t+1)},\phi^{(t+1)})
 =(\theta^{(t)},\phi^{(t)})-\eta\nabla\mathcal L^{(t)},
 \qquad
 \mathcal L^{(t)}:=\mathcal L(\theta^{(t)},\phi^{(t)}),
\end{equation}
where $\nabla\mathcal L^{(t)}$ denotes the gradient evaluated at that iterate. For every $t\geq0$,
\begin{equation}\label{eqn:bp-descent}
 \mathcal L^{(t+1)}\leq\mathcal L^{(t)}
       -\frac\eta2\|\nabla\mathcal L^{(t)}\|^2.
\end{equation}
In particular, the loss converges, $\|\nabla\mathcal L^{(t)}\|\to0$, and, for $t\geq1$,
\begin{equation}\label{eqn:bp-stationarity}
 \min_{0\leq j<t}\|\nabla\mathcal L^{(j)}\|^2
 \leq\frac{2\mathcal L^{(0)}}{\eta t}.
\end{equation}
Under the uniform initialization of Theorem~\ref{thm:bp-loss}, the expected decrease after the first update also obeys
\begin{align}\label{eqn:bp-firstdescent}
 \mathbb E_{\theta,\phi}[\mathcal L^{(0)}-\mathcal L^{(1)}]
 \geq\frac{\eta(2L_N+1)\sum_\mu|a_\mu|^4}{128\sqrt{2(L_N+1)}}
 \left[\left(\frac1m\sum_i y_i\right)^2
       +2\left|\frac1m\sum_i y_i\tr\bigl(\rho_SU_N^{(i)}\bigr)\right|^2\right].
\end{align}
The expectation here is over the initial angles only.
\end{proposition}
\begin{proof}
All gates are unitary. A first derivative of a rotation gate has operator norm $1/2$, and a second derivative of a gate, or a mixed derivative at two gates in the same branch, has norm $1/4$. The product rule applied to $r_\mu(x)=\bra0Q_\mu(x)^\dagger ZQ_\mu(x)\ket0$ therefore gives
\begin{equation*}
 |\partial_{\alpha_{l,\mu}}r_\mu(x)|\leq1,\qquad
 |\partial_{\alpha_{l,\mu}}\partial_{\beta_{j,\mu}}r_\mu(x)|\leq1.
\end{equation*}
Spectral averaging preserves these bounds, and Equation~\eqref{eqn:bp-reduction} supplies the factor $|a_\mu|^2$. Mixed second derivatives from distinct branches vanish. Each branch contributes $2L_N+1$ coordinates to the gradient vector, so
\begin{align}\label{eqn:bp-readoutderivativebounds}
 \|\nabla f_Z(\bfU^{(i)})\|^2
 &\leq(2L_N+1)\sum_\mu|a_\mu|^4,\nonumber\\
 \|\nabla^2 f_Z(\bfU^{(i)})\|_{\mathrm{op}}
 &\leq(2L_N+1)\max_\mu|a_\mu|^2.
\end{align}
The Hessian estimate follows from its block diagonal form and the absolute row sum bound within each symmetric block. Since $|f_Z(\bfU^{(i)})-y_i|\leq2$, differentiation of Equation~\eqref{eqn:learning-loss} yields
\begin{align}\label{eqn:bp-smoothness}
 \|\nabla^2\mathcal L\|_{\mathrm{op}}
 &\leq\frac1m\sum_i\left(
       \|\nabla f_Z(\bfU^{(i)})\|^2
       +2\|\nabla^2 f_Z(\bfU^{(i)})\|_{\mathrm{op}}\right)\nonumber\\
 &\leq3(2L_N+1)\max_\mu|a_\mu|^2,
\end{align}
where $\sum_\mu|a_\mu|^4\leq\max_\mu|a_\mu|^2$ uses $\sum_\mu|a_\mu|^2=1$. This bound holds at all angles. Taylor's formula with its integral remainder along the update direction gives
\begin{equation*}
 \mathcal L^{(t+1)}
 \leq\mathcal L^{(t)}-\eta\|\nabla\mathcal L^{(t)}\|^2
       +\frac{3\eta^2(2L_N+1)\max_\mu|a_\mu|^2}{2}
                    \|\nabla\mathcal L^{(t)}\|^2.
\end{equation*}
Equation~\eqref{eqn:bp-stepsize} proves Equation~\eqref{eqn:bp-descent}. Summing it and using $\mathcal L^{(t)}\geq0$ gives
\begin{equation*}
 \sum_{j=0}^{t-1}\|\nabla\mathcal L^{(j)}\|^2
 \leq\frac{2\mathcal L^{(0)}}\eta.
\end{equation*}
This proves Equation~\eqref{eqn:bp-stationarity}. The series remains bounded as $t\to\infty$, so its nonnegative terms tend to zero. The loss is decreasing and bounded below, hence convergent.

Finally, apply Equation~\eqref{eqn:bp-descent} at $t=0$ and average over the initial angles. The expected squared gradient norm is the sum of the coordinate second moments. These equal their variances by Equation~\eqref{eqn:bp-lossmean}. Sum Equation~\eqref{eqn:bp-lossbound} over the $2L_N+1$ listed coordinates in every branch with $a_\mu\neq0$. Branches of zero weight contribute zero. This gives Equation~\eqref{eqn:bp-firstdescent}.
\end{proof}

These statements guarantee loss decrease and convergence to stationarity. They do not determine the limiting loss. A sufficient condition for fitting every training target is a uniform lower bound on the sample Jacobian along the optimization trajectory. This is an additional condition on the data and iterates.

\begin{corollary}[Conditional convergence to zero training loss]\label{cor:bp-fitting}
For the iterates of Proposition~\ref{prop:bp-descent}, let $J^{(t)}$ be the real $m\times d_{\cI}(2L_N+1)$ matrix whose entry in sample row $i$ and parameter column $(\alpha,l,\mu)$, in the ordering of $\nabla\mathcal L$, is
\begin{equation}\label{eqn:bp-jacobian}
 J^{(t)}_{i,(\alpha,l,\mu)}
 =\left.\partial_{\alpha_{l,\mu}}f_Z(\bfU^{(i)})
       \right|_{(\theta,\phi)=(\theta^{(t)},\phi^{(t)})}.
\end{equation}
If there is a constant $\lambda>0$ such that
\begin{equation}\label{eqn:bp-jacobiancondition}
 \frac1mJ^{(t)}(J^{(t)})^{\mathsf T}\succeq\lambda I_m
 \qquad\text{for every }t\geq0,
\end{equation}
then
\begin{equation}\label{eqn:bp-fitting}
 \mathcal L^{(t)}\leq(1-\eta\lambda)^t\mathcal L^{(0)}
 \longrightarrow0.
\end{equation}
\end{corollary}
\begin{proof}
Write $f_Z^{(t)}$ for the prediction evaluated at the $t$th iterate. Equations~\eqref{eqn:bp-lossgradient} and~\eqref{eqn:bp-jacobiancondition} imply
\begin{align*}
 \|\nabla\mathcal L^{(t)}\|^2
 &=\frac1{m^2}\sum_{i,j}
       [f_Z^{(t)}(\bfU^{(i)})-y_i]
       [J^{(t)}(J^{(t)})^{\mathsf T}]_{ij}
       [f_Z^{(t)}(\bfU^{(j)})-y_j]\\
 &\geq\frac\lambda m\sum_i[f_Z^{(t)}(\bfU^{(i)})-y_i]^2
 =2\lambda\mathcal L^{(t)}.
\end{align*}
Substitution into Equation~\eqref{eqn:bp-descent} gives
$\mathcal L^{(t+1)}\leq(1-\eta\lambda)\mathcal L^{(t)}$.
Also, Equation~\eqref{eqn:bp-readoutderivativebounds} bounds the operator norm of $J^{(t)}(J^{(t)})^{\mathsf T}/m$ by $(2L_N+1)\max_\mu|a_\mu|^2$, so $0<\eta\lambda\leq1/3$. Iteration proves Equation~\eqref{eqn:bp-fitting}.
\end{proof}

The condition in Equation~\eqref{eqn:bp-jacobiancondition} is not a consequence of positive coordinate variances at initialization. In particular, all predictions lie in the real span of the $2L_N+1$ sample features $1$, $\operatorname{Re}\tr(\rho_S(U_N^{(i)})^h)$ and $\operatorname{Im}\tr(\rho_S(U_N^{(i)})^h)$ for $1\leq h\leq L_N$. Their parameter derivatives lie in the same span. Hence $\operatorname{rank}J^{(t)}\leq2L_N+1$, and the stated full-row-rank condition cannot hold when $m>2L_N+1$. Establishing the Jacobian bound requires a separate analysis of the chosen task and trajectory.

There is also a direct obstruction to an unconditional fitting claim. Consider two samples with the same $U_N$ and targets $1$ and $-1$. They may differ in $U_1,\ldots,U_{N-1}$ when $N\geq2$. Equation~\eqref{eqn:bp-reduction} forces equal predictions, so at every parameter value,
\begin{equation}\label{eqn:bp-fittingobstruction}
 \mathcal L=\frac14\left[(f_Z-1)^2+(f_Z+1)^2\right]
 =\frac12(f_Z^2+1)\geq\frac12.
\end{equation}
The readout-gradient lower bound still holds. Thus the readout variance, the squared-loss variance, descent to stationarity and convergence to zero training loss are distinct conclusions with distinct assumptions.
\section{Exact scalar synthesis at a fixed query order}\label{sec:sdp}

This appendix gives an exact characterization for scalar targets when the chronological order of the queries is fixed.  A semidefinite program (SDP) determines the smallest normalization when arbitrary input-independent unitaries and finite workspace are allowed.  All definitions needed for this characterization are introduced below.

\subsection{Coefficient histories}

Fix $N\geq1$ scalar inputs $\bfx\in\RR^N$ and a chronological query word $\chi=(j_1,\ldots,j_K)$, where $K\geq0$ and $j_t\in\set{1,\ldots,N}$.  Let $L_j$ be the number of occurrences of $j$ in $\chi$.  The relevant target space is
\begin{equation}\label{eqn:sdp-space}
 \cE_{\bfL}:=\operatorname{span}\setcond{\exp\!\left[\frac{i}{2}(2\bfn-\bfL)\mathbin{\cdot}\bfx\right]}{0\leq n_j\leq L_j\ \text{for every }j}.
\end{equation}
For $P\in\cE_{\bfL}$, write $\Gamma_\chi(P)$ for the infimum of the positive numbers $\lambda$ for which a finite workspace with designated query qubits $B_1,\ldots,B_N$ and input-independent unitaries $V_0,\ldots,V_K$ satisfies
\begin{equation}\label{eqn:sdp-model}
 \bra0 V_KR_z(x_{j_K})_{B_{j_K}}V_{K-1}\cdots V_1R_z(x_{j_1})_{B_{j_1}}V_0\ket0
 =\frac{P(\bfx)}{\lambda}
 \qquad\text{for every }\bfx.
\end{equation}
Here $R_z(x)=\operatorname{diag}(e^{-ix/2},e^{ix/2})$.  The rightmost factor acts first.  There is no intermediate measurement or postselection.  We set $\Gamma_\chi(0)=0$.

Put $z_j=e^{ix_j}$ and remove the common phase of the queries by writing
\begin{equation}\label{eqn:sdp-centered}
 \begin{aligned}
 p(\bfz)&:=e^{i\bfL\cdot\bfx/2}P(\bfx)
       =\sum_{\bfn\in I_K}c_{\bfn}\bfz^{\bfn},
 & \bfz^{\bfn}&:=\prod_{j=1}^Nz_j^{n_j},\\
 I_t&:=\prod_{j=1}^N\set{0,\ldots,\ell_j(t)},
 & \ell_j(t)&:=\abs{\setcond{s\leq t}{j_s=j}}.
 \end{aligned}
\end{equation}
Here $q_t:=|I_t|$, $Q:=q_K=\prod_j(L_j+1)$, and $c=(c_{\bfn})_{\bfn\in I_K}\in\CC^{1\times Q}$.  A variable with $L_j=0$ simply contributes a singleton factor to every $I_t$.  Let $e_{\bfn}$ denote a coefficient-space coordinate vector and let $e_j$ denote the $j$-th coordinate vector of $\ZZ^N$.  For $b\in\set{0,1}$, define the known injections
\begin{equation}\label{eqn:sdp-injections}
 E_{t,b}:\CC^{I_{t-1}}\longrightarrow\CC^{I_t},
 \qquad E_{t,b}e_{\bfn}=e_{\bfn+b e_{j_t}}.
\end{equation}
Each $E_{t,b}$ is an isometry.  Their images can overlap: after repeated queries to one variable, different bit strings can contribute to the same monomial.  We write $X\succeq Y$ when $X-Y$ is positive semidefinite.

In the program below, $s$ is real, each $G_t$ is Hermitian of size $q_t\times q_t$, and each $H_{t,b}$ is Hermitian of size $q_{t-1}\times q_{t-1}$.

\begin{theorem}[Exact compilation at a fixed query order]\label{thm:sdp-primal}
Let $P\in\cE_{\bfL}$ and let $\chi$ be fixed.  A circuit realizes $P/\lambda$ in the sense of Equation~\eqref{eqn:sdp-model}, for $\lambda>0$, if and only if the following program is feasible with $s=\lambda^2$:
\begin{equation}\label{eqn:sdp-primal}
 \begin{aligned}
 \Gamma_\chi(P)^2=\min\quad &s\\
 \text{subject to}\quad
 &G_0=[s],\qquad s\geq0,\\
 &H_{t,0},H_{t,1}\succeq0,\qquad H_{t,0}+H_{t,1}\preceq G_{t-1}
 &&(1\leq t\leq K),\\
 &G_t=\sum_{b=0}^1E_{t,b}H_{t,b}E_{t,b}^{\dagger}
 &&(1\leq t\leq K),\\
 &G_K\succeq c^{\dagger}c.
 \end{aligned}
\end{equation}
The minimum is attained.  For $P\neq0$, a minimizing solution recovers a circuit at scale $\Gamma_\chi(P)$ by Gram factorization and unitary completion.
\end{theorem}

\begin{proof}
\textit{Extracting a feasible history.}
Suppose first that Equation~\eqref{eqn:sdp-model} holds at scale $\lambda$.  After the first $t$ queries, and before the following known unitary, expand the scaled and centered state in its monomials.  Its coefficient matrix is denoted $A_t$, with one physical-workspace row per basis state and one column per element of $I_t$; set $A_0=\lambda\ket0$.  Let
\begin{equation}\label{eqn:sdp-extract}
 B_{t,b}:=\Pi_{j_t,b}V_{t-1}A_{t-1},
 \qquad H_{t,b}:=B_{t,b}^{\dagger}B_{t,b},
 \qquad G_t:=A_t^{\dagger}A_t,
\end{equation}
where $\Pi_{j,b}$ projects onto the $b$ eigenspace of query qubit $B_j$.  Orthogonality of the two eigenspaces gives $H_{t,0}+H_{t,1}=G_{t-1}$.  The query shifts the two branches according to Equation~\eqref{eqn:sdp-injections}, so
\begin{equation}\label{eqn:sdp-update}
 A_t=\sum_{b=0}^1B_{t,b}E_{t,b}^{\dagger},
 \qquad
 G_t=\sum_{b=0}^1E_{t,b}H_{t,b}E_{t,b}^{\dagger}.
\end{equation}
After $V_K$, the selected output row is $c$ by Fourier uniqueness.  The remaining output rows have some coefficient matrix $R$, whence $G_K=c^{\dagger}c+R^{\dagger}R$.  This gives a feasible point with $s=\lambda^2$.

\textit{Removing branch slack.}
Conversely, start from any feasible point with $s=\lambda^2>0$.  Define $G'_0=G_0$ and, recursively,
\begin{equation}\label{eqn:sdp-lossless-recovery}
 \begin{aligned}
 R_t&:=G'_{t-1}-H_{t,0}-H_{t,1},\\
 H'_{t,0}&:=H_{t,0}+R_t,\qquad H'_{t,1}:=H_{t,1},\\
 G'_t&:=\sum_{b=0}^1E_{t,b}H'_{t,b}E_{t,b}^{\dagger}.
 \end{aligned}
\end{equation}
Inductively, $G'_{t-1}\succeq G_{t-1}$, hence $R_t\succeq0$, $G'_t\succeq G_t$, and $H'_{t,0}+H'_{t,1}=G'_{t-1}$.  In particular, the terminal condition remains true.  We may therefore replace the original history by this lossless one and drop the primes.

\textit{Recovering the gates.}
Choose an allowable workspace of dimension $D=2^Nd_A\geq2Q$.  Factor each $H_{t,b}=B_{t,b}^{\dagger}B_{t,b}$ and place the rows of $B_{t,b}$ in the $b$ eigenspace of query $j_t$.  This is possible because each factor has rank at most $Q\leq D/2$.  The two factors have orthogonal physical supports and, by losslessness,
\begin{equation}\label{eqn:sdp-gram-match}
 (B_{t,0}+B_{t,1})^{\dagger}(B_{t,0}+B_{t,1})=G_{t-1}.
\end{equation}
If two matrices have the same Gram matrix, the assignment from the columns of one to the corresponding columns of the other is a well-defined isometry on their column spans.  It extends to a unitary on the finite workspace.  This argument also applies when either Gram matrix is rank deficient.  Beginning with $A_0=\lambda\ket0$, use this unitary as $V_{t-1}$, apply the prescribed query, and use Equation~\eqref{eqn:sdp-update} to obtain a coefficient matrix with Gram $G_t$.  Iterating gives $A_K$.

Finally, factor $G_K-c^{\dagger}c=R^{\dagger}R$ using at most $Q$ rows.  The matrix obtained by stacking the row $c$ above $R$ has Gram matrix $G_K$ and fits in $D\geq2Q$ rows.  One final unitary maps $A_K$ to this matrix, with $c$ as the selected output row.  Restoring the common phase proves Equation~\eqref{eqn:sdp-model}.

\textit{Attainment.}
For $0<\delta<1/2$, set $G_0=[s]$ with $s>0$, take $H_{t,b}=\delta G_{t-1}$, and use the update in Equation~\eqref{eqn:sdp-primal}.  The union of the images of $E_{t,0}$ and $E_{t,1}$ covers every coordinate of $I_t$, so every $G_t$ is positive definite.  Scaling $s$ makes $G_K-c^{\dagger}c$ positive definite as well.  Thus the program has a strictly feasible point.  Moreover,
\begin{equation}\label{eqn:sdp-compact}
 \tr G_t=\tr(H_{t,0}+H_{t,1})\leq\tr G_{t-1}\leq s.
\end{equation}
Positive semidefiniteness and these trace bounds make every bounded objective level set compact, proving that the minimum is attained.  If $P\neq0$, then $s\geq\tr(c^{\dagger}c)>0$, so the recovery argument applies to a minimizer.  If $P=0$, the optimum is zero, while the same recovery works at every positive scale.  When $K=0$, the program is simply $s\geq|c|^2$, and a two-dimensional unitary on one query qubit, with all other designated qubits as spectators, realizes every positive feasible scale.
\end{proof}

\subsection{Dual certificates and physical workspace}

\begin{theorem}[Matching dual certificate]\label{thm:sdp-dual}
For the coefficient row $c$ and the injections in Equation~\eqref{eqn:sdp-injections}, the dual of Equation~\eqref{eqn:sdp-primal} is
\begin{equation}\label{eqn:sdp-dual}
 \begin{aligned}
 \Gamma_\chi(P)^2=\max\quad &cY_Kc^{\dagger}\\
 \text{subject to}\quad
 &Y_0=[1],\qquad Y_t\succeq0 &&(1\leq t\leq K),\\
 &E_{t,b}^{\dagger}Y_tE_{t,b}\preceq Y_{t-1}
 &&(1\leq t\leq K,\ b\in\set{0,1}).
 \end{aligned}
\end{equation}
The maximum is attained.  Thus every feasible sequence $(Y_0,\ldots,Y_K)$ proves $\Gamma_\chi(P)\geq\sqrt{cY_Kc^{\dagger}}$.
\end{theorem}

\begin{proof}
We first identify the displayed program as the Lagrange dual.  Suppose $K\geq1$, introduce positive semidefinite slacks $S_t:=G_{t-1}-H_{t,0}-H_{t,1}$ and $T:=G_K-c^{\dagger}c$, and substitute the Gram updates.  The constraint $s\geq0$ is then redundant, so $s$ may be minimized as a free scalar.  With $\langle A,B\rangle:=\tr(AB)$, a backward choice of Hermitian multipliers gives the Lagrangian
\begin{align*}
 \mathcal L={}&s-\left\langle Y_0,[s]-H_{1,0}-H_{1,1}-S_1\right\rangle\\
 &-\sum_{t=2}^K\left\langle Y_{t-1},
     \sum_{b=0}^1E_{t-1,b}H_{t-1,b}E_{t-1,b}^{\dagger}
     -H_{t,0}-H_{t,1}-S_t\right\rangle\\
 &-\left\langle Y_K,
     \sum_{b=0}^1E_{K,b}H_{K,b}E_{K,b}^{\dagger}-c^{\dagger}c-T\right\rangle .
\end{align*}
Minimization over $s$, the positive semidefinite $H_{t,b}$, and the slacks is finite exactly when $Y_0=[1]$, $Y_t\succeq0$, and $E_{t,b}^{\dagger}Y_tE_{t,b}\preceq Y_{t-1}$.  The remaining constant is $cY_Kc^{\dagger}$, which proves that Equation~\eqref{eqn:sdp-dual} is the Lagrange dual.  For $K=0$, the same dual follows directly from $s\geq|c|^2$.

For any primal and dual feasible points, positivity of $Y_K$ and the terminal constraint give $cY_Kc^{\dagger}\leq\tr(G_KY_K)$.  Repeatedly applying the update and the two compression inequalities gives
\begin{equation}\label{eqn:sdp-weak-duality}
 \tr(G_tY_t)
 \leq\tr(G_{t-1}Y_{t-1})
 \leq\cdots\leq\tr(G_0Y_0)=s.
\end{equation}
This proves weak duality.  The strictly feasible primal point constructed in the proof of Theorem~\ref{thm:sdp-primal}, together with finite primal optimum, gives strong duality and dual attainment by Slater's theorem.  For $K=0$, both programs have value $|c|^2$ directly.
\end{proof}

For every $\bfz$ on the unit torus, let $m_t(\bfz)=(\bfz^{\bfn})_{\bfn\in I_t}$.  The choice $Y_t=m_t(\bfz)m_t(\bfz)^{\dagger}$ is dual feasible because $E_{t,b}^{\dagger}m_t(\bfz)=z_{j_t}^bm_{t-1}(\bfz)$.  It gives $cY_Kc^{\dagger}=|P(\bfx)|^2$, and hence $\Gamma_\chi(P)\geq\norm{P}_\infty$.  Other dual histories can give a stricter lower bound.

\begin{proposition}[Exact fixed-workspace criterion]\label{prop:sdp-workspace}
Fix a scale $\lambda>0$ and a workspace of dimension $D=2^Nd_A$, consisting of the $N$ designated query qubits and an auxiliary register of dimension $d_A$.  A circuit in this workspace realizes $P/\lambda$ if and only if Equation~\eqref{eqn:sdp-primal} has a feasible history with $G_0=[\lambda^2]$, with the branch inequalities replaced by the lossless equalities
\begin{equation}\label{eqn:sdp-lossless}
 H_{t,0}+H_{t,1}=G_{t-1}\qquad(1\leq t\leq K),
\end{equation}
and with
\begin{equation}\label{eqn:sdp-rank}
 \operatorname{rank}H_{t,b}\leq D/2\quad\text{for every }t,b,
 \qquad
 \operatorname{rank}(G_K-c^{\dagger}c)\leq D-1.
\end{equation}
Every feasible unrestricted normalization can be realized, for example, in a qubit workspace of dimension
\begin{equation}\label{eqn:sdp-workspacebound}
 D=2^{\max\set{N,\ceil{\log_2(2Q)}}}.
\end{equation}
\end{proposition}

\begin{proof}
A circuit gives a lossless history by Equation~\eqref{eqn:sdp-extract}.  Each branch factor lies in a $D/2$-dimensional query eigenspace, proving the first rank bound.  Removing the selected output row leaves $D-1$ rows with Gram matrix $G_K-c^{\dagger}c$, proving the second.

Conversely, the branch rank bounds place the factors in the two physical eigenspaces, Equation~\eqref{eqn:sdp-lossless} gives the Gram equality needed for each unitary completion, and the terminal rank bound places the residual below the selected output row.  This is exactly the recovery construction in the proof of Theorem~\ref{thm:sdp-primal}, now within the specified dimension.  All relevant ranks are at most $Q$, so every permitted $D\geq2Q$ suffices; Equation~\eqref{eqn:sdp-workspacebound} is one such choice.
\end{proof}

The unrestricted characterization is a convex SDP.  The fixed-workspace statement adds the rank constraints in Equation~\eqref{eqn:sdp-rank}, so it is an exact nonconvex feasibility criterion rather than another SDP.

\subsection{An order-dependent normalization}

The query counts alone do not determine the optimum.  Let $H=2^{-1/2}\begin{psmallmatrix}1&1\\1&-1\end{psmallmatrix}$ and $\ket+=2^{-1/2}(\ket0+\ket1)$, and define
\begin{equation}\label{eqn:sdp-orderfunction}
 \begin{aligned}
 f(a,b,c)&:=\bra+R_z(c)H R_z(b)H R_z(a)\ket+\\
 &=\cos(a/2)\cos(c/2)e^{-ib/2}
   -\sin(a/2)\sin(c/2)e^{ib/2}.
 \end{aligned}
\end{equation}

\begin{proposition}[Equal query counts can require different normalizations]\label{prop:sdp-order}
For the chronological words $(a,b,c)$ and $(a,c,b)$, each with one query to each variable,
\begin{equation}\label{eqn:sdp-ordergap}
 \norm{f}_\infty=\Gamma_{(a,b,c)}(f)=1,
 \qquad
 \Gamma_{(a,c,b)}(f)=\sqrt{2},
 \qquad
 \norm{f}_{c,1}=2.
\end{equation}
The separation holds even with unrestricted finite workspace.
\end{proposition}

\begin{proof}
The first expression in Equation~\eqref{eqn:sdp-orderfunction} is a unitary matrix element.  It equals one at $a=b=c=0$, so it proves $\norm{f}_\infty=\Gamma_{(a,b,c)}(f)=1$.  For the word $(a,c,b)$, initialize two qubits in $\ket{00}$, apply $R_x(a)=HR_z(a)H$ to the first and $R_x(c)$ to the second, then apply $R_z(b)$ to the first.  Projection onto $(\ket{00}+\ket{11})/\sqrt2$ gives $f/\sqrt2$.  Known SWAP gates can route the active qubits through the designated query registers.

For the converse, cut an arbitrary normalized circuit immediately before its final, $b$ query.  Absorb all earlier known gates into a unit vector $\psi(a,c)$, and write the final row as
\begin{equation}\label{eqn:sdp-orderrow}
 r(b)=r_-e^{-ib/2}+r_+e^{ib/2},
 \qquad
 \norm{r_-}_2^2+\norm{r_+}_2^2=1.
\end{equation}
If this circuit realizes $f/\lambda$, Fourier uniqueness in $b$ gives
\begin{align*}
 r_-\psi(a,c)&=\frac{\cos(a/2)\cos(c/2)}{\lambda},
 &r_+\psi(a,c)&=-\frac{\sin(a/2)\sin(c/2)}{\lambda}.
\end{align*}
Both coefficient functions have supremum norm one.  Cauchy--Schwarz and $\norm{\psi(a,c)}_2=1$ therefore give $\norm{r_-}_2,\norm{r_+}_2\geq1/\lambda$.  Equation~\eqref{eqn:sdp-orderrow} implies $\lambda^2\geq2$.

The SDP also has an exact dual certificate for this lower bound.  For the word $(a,c,b)$, centering gives
\begin{equation}\label{eqn:sdp-ordercentered}
 p_f(\bfz):=e^{i(a+b+c)/2}f(a,b,c)
 =\frac14\bigl[(1+z_a)(1+z_c)+z_b(1-z_a)(1-z_c)\bigr].
\end{equation}
Its coefficient row is $c_f=w^{\mathsf T}/4$, where
\begin{equation}\label{eqn:sdp-orderwitness}
 w_{n_a,n_b,n_c}=(-1)^{n_b(n_a+n_c)}
 \qquad ((n_a,n_b,n_c)\in\set{0,1}^3).
\end{equation}
Thus $\norm{f}_{c,1}=2$.  Define
\begin{equation}\label{eqn:sdp-orderdual}
 \begin{aligned}
 Y_3&=\frac12ww^{\mathsf T},\\
 Y_2&=\frac12(v_0v_0^{\mathsf T}+v_1v_1^{\mathsf T}),
 &v_0&=(1,1,1,1)^{\mathsf T},\quad v_1=(1,-1,-1,1)^{\mathsf T},\\
 Y_1&=I_2,\qquad Y_0=[1].
 \end{aligned}
\end{equation}
Here $Y_2$ is indexed by $(n_a,n_c)=(0,0),(0,1),(1,0),(1,1)$.  The two compressions of $Y_3$ through the final $b$ query sum to $Y_2$ and are each bounded by it.  Each compression of $Y_2$ through the $c$ query is $I_2$, and each compression of $Y_1$ through the $a$ query is $[1]$.  Hence this is dual feasible, and
\begin{equation}\label{eqn:sdp-ordervalue}
 c_fY_3c_f^{\dagger}=\frac1{32}(w^{\mathsf T}w)^2=2.
\end{equation}
Together with the explicit circuit, this certifies the exact optimum without a numerical SDP solve.
\end{proof}

The same target has a terminal polynomial completion at scale one: direct expansion shows $|p_f(\bfz)|^2+|q(\bfz)|^2=1$ on the unit torus for
\begin{equation}\label{eqn:sdp-ordercompletion}
 q(\bfz)=\frac14\bigl[(1+z_a)(1-z_c)+z_b(1-z_a)(1+z_c)\bigr].
\end{equation}
Both polynomials have the allowed support.  Thus a terminal completion and a pointwise norm bound alone do not ensure sequential realization in a prescribed query order.
\end{document}